\documentclass[10pt,a4paper]{article}
\usepackage[T1]{fontenc}
\usepackage[utf8]{inputenc}
\usepackage{amsmath}
\usepackage{amsfonts}
\usepackage{amssymb}
\usepackage{amsthm}
\usepackage[square,numbers]{natbib}
\usepackage{hyperref} 
\usepackage[a4paper, total={6.5in, 8in}]{geometry}
\usepackage[dvipsnames, table]{xcolor}
\usepackage{graphicx}
\usepackage{float}
\usepackage{doi}
\hypersetup{hidelinks}
\usepackage{tikz}
\usepackage[export]{adjustbox}
\usepackage{layouts}
\usepackage{psfrag}
\usepackage{stmaryrd}
\usepackage{comment}
\usepackage{booktabs}
\usetikzlibrary{shapes.geometric, arrows}

\newcommand{\nn}{\mathcal{NN}}

\newcommand{\relu}{\textnormal{ReLu}}
\newtheorem{theorem}{Theorem}
\newtheorem{corollary}[theorem]{Corollary}
\newtheorem{lemma}[theorem]{Lemma}

\newtheorem{remark}{Remark}
\usepackage{mathtools}
\usepackage{multicol}
\usepackage{pgfplots}
\usetikzlibrary{arrows,shapes,positioning}
\usepgfplotslibrary{groupplots}
\usepackage{multirow}
\usepackage{caption}
\usepackage{subcaption}

\newcommand{\R}{\mathbb{R}}

\newcommand{\SO}{\mathcal{SO}}

\newcommand{\Perm}{\operatorname{Perm}}

\newcommand{\m}{m}
\newcommand{\bsym}{\boldsymbol}
\usepackage{colortbl}
\usepackage[hang,flushmargin,symbol]{footmisc}

\begin{document}

\begin{center}
\noindent\textbf{\Large{Modeling isotropic polyconvex hyperelasticity by neural networks -- sufficient and necessary criteria for compressible and incompressible materials}}
~\\
~\\
~\\
Gian-Luca Geuken$^1$, Patrick Kurzeja$^1$, David Wiedemann$^2$, Martin Zlati\'{c}$^3$, Marko \v{C}ana\dj{}ija$^3$, J\"orn Mosler$^{1,}$\footnote{\noindent Corresponding author\\Email adress: joern.mosler@tu-dortmund.de}
~\\
~\\
$^1$\textit{Institute of Mechanics, Department of Mechanical Engineering,
TU Dortmund University, Leonhard-Euler-Str. 5, 44227 Dortmund, Germany}\\
$^2$\textit{Applied Analysis, Faculty of Mathematics,
TU Dortmund University, Vogelpothsweg 87, 44227 Dortmund, Germany}\\
$^3$\textit{Faculty of Engineering, University of Rijeka, Vukovarska 58, 51000 Rijeka, Croatia}
~\\
~\\
\end{center}
Abstract:
This work investigates different sufficient and necessary criteria for hyperelastic, isotropic polyconvex material models, focusing on neural network implementations for compressible and incompressible materials. Furthermore, the expressiveness, accuracy, simplicity as well as the efficiency of those models is analyzed. This also enables an assessment of the practical applicability of the models. Convex Signed Singular Value Neural Networks (CSSV-NNs) are applied to compressible materials and tailored to incompressibility (inc-CSSV-NNs), resulting in a universal approximation for frame-indifferent, isotropic polyconvex energies for the compressible as well as incompressible case. While other existing approaches also guarantee frame-indifference, isotropy and polyconvexity, they impose too restrictive constraints and thus limit the expressiveness of the model. This is further substantiated by numerical examples of several, well-established classical models (Neo--Hooke, Mooney--Rivlin, Gent and Arruda--Boyce) and Treloar's experimental data. Moreover, the numerical examples include an explicitly constructed energy function that cannot be approximated by neural networks constrained by an improved version of Ball’s criterion for polyconvexity. This substantiates that Ball’s criterion, though sufficient, is not necessary for polyconvexity.
~
\thispagestyle{empty}
\clearpage

 \section{Introduction}

It is well known that embedding physical and mathematical principles -- such as frame-indifference, isotropy and polyconvexity -- into neural network architectures can significantly improve model robustness and generalization \citep{kumar2022,peng2021,moseley2022,geuken2024}. However, these constraints often come at the cost of expressiveness: enforcing physical properties too rigidly may restrict the representable function space more than necessary. Thus, the central challenge lies in designing machine learning models that incorporate physical knowledge without over-constraining the system.

In the context of hyperelastic materials, polyconvexity serves as a key condition to ensure material stability and the existence of solutions. Traditionally, frame-indifferent and isotropic polyconvex energies are formulated using the invariants of the right Cauchy--Green tensor or the principal stretches, following the seminal polyconvexity criterion by \citet{ball1976}. While this classical framework is mathematically sound and practically convenient, it is sufficient but not necessary for polyconvexity \citep{wiedemann2026}, which inherently limits model flexibility. This limitation directly transfers to neural network models building upon this framework, where polyconvexity improves the robustness, generalization and extrapolation behavior of the model.

Due to this observation, \citet{geuken2025} introduced the Convex Signed Singular Value Neural Network (CSSV-NN) architecture, which is based on a necessary and sufficient criterion for frame-indifferent, isotropic polyconvex energies derived by \citet{wiedemann2023, wiedemann2026}. They can universally approximate any such energy function to arbitrary precision. In contrast, many existing approaches \citep{shen2004,liang2008,linka2021,asad2022,chen2022,linden2023,klein2023,balazi2025,vijay2025,abdolazizi2025} either fail to strictly enforce frame-indifference, isotropy and polyconvexity or do so under overly restrictive sufficient criteria, thereby limiting their expressiveness as stated above.

This work systematically analyzes sufficient and necessary criteria for hyperelastic, frame-indifferent, isotropic and polyconvex energies and investigates their implications for neural network models. This also enables an assessment of the practical applicability of the models. The CSSV-NNs are further extended and adapted to the incompressible case (denoted by inc-CSSV-NN). This is particularly relevant for synthetic materials such as polymers \citep{treloar1944, arrudaboyce1993,steinmann2012} and biological materials such as the periodontal ligament \citep{rees1997,kurzeja2025}. As for the compressible case, existing formulations of frame-indifferent, isotropic polyconvex energies for incompressibility (e.g., \citep{pierre2023,zlatic2024,holthusen2024,dammass2025,kalina2025,holthusen2026}) are too restrictive or do not strictly enforce frame-indifference, isotropy and polyconvexity.

To guide the reader, the main contributions and focus points of this work are summarized as follows:
\begin{itemize}
    \item Analysis of sufficient and necessary criteria for frame-indifferent, isotropic polyconvex hyperelasticity in the compressible and incompressible case.
    \item Extension and adaptation of Convex Signed Singular Value Neural Networks (CSSV-NNs) to incompressible materials (inc-CSSV-NNs).
    \item Proof of a universal approximation theorem for inc-CSSV-NNs representing frame-indifferent, isotropic polyconvex energies of incompressible materials.
    \item Comparative evaluation of expressiveness, accuracy, simplicity as well as efficiency of the presented neural networks for several, well-established models (Neo--Hooke, Mooney--Rivlin, Gent, Arruda--Boyce) and Treloar’s experimental data.
\end{itemize}
In addition numerical examples are provided that
\begin{itemize}
    \item substantiate that Ball’s criterion is sufficient but not necessary for polyconvexity in three dimensions,
    \item indicate that the reduced criterion in terms of the elementary symmetric polynomials of the signed singular values from \citep{wiedemann2026} is also not necessary for polyconvexity in three dimensions,
\end{itemize}
motivating future research in those directions.

The paper is organized as follows. Section 2 constitutes physical and mathematical constraints for hyperelastic modeling and points out how to achieve frame-indifference, isotropy and incompressibility. Based on that, four different criteria for frame-indifferent, isotropic polyconvex energies of compressible and incompressible materials are presented in Section 3. Those are then used for designing appropriate neural networks for incompressible material behavior in Section 4. Section 5 provides a proof of the universal approximation theorem of inc-CSSV-NNs for frame-indifferent, isotropic polyconvex functions of incompressible materials. The algorithmic implementation in terms of network activation, hyperparameters and training is described in Section 6. Finally, numerical examples of several, well-established classical models and experimental data, such as Treloar's data, are discussed in Section 7. Moreover, the numerical examples include an explicitly constructed energy function that cannot be approximated by neural networks constrained by an improved version of Ball’s criterion for polyconvexity.

\section{Physical and mathematical constraints for hyperelastic modeling} \label{sec_fundamentals}
Hyperelastic material models should respect fundamental physical principles and mathematical restrictions. In this work, the focus is on compressible and incompressible, isotropic hyperelastic materials.
\subsection{The compressible case}
From a physical standpoint, a hyperelastic model should satisfy the following requirements \citep{truesdell1965, geuken2025}:
\begin{enumerate}
\item \textbf{Path Independence:}  
Stresses derive from a potential $\Psi$. For example, the first Piola--Kirchhoff stress tensor is given by
\begin{equation}
\boldsymbol{P} = \frac{\partial \Psi}{\partial \boldsymbol{F}}, \label{eq:definition_P}
\end{equation}
where $\boldsymbol{F}$ denotes the deformation gradient.

\item \textbf{Thermodynamic Consistency:}  
The second law of thermodynamics must hold (which is automatically satisfied in hyperelasticity if stresses derive from $\Psi$), reading here:
\begin{equation}
\boldsymbol{P} : \dot{\boldsymbol{F}} - \dot{\Psi} = 0.
\end{equation}

\item \textbf{Frame Indifference (Objectivity):}  
The energy must remain invariant under superimposed rigid body rotations of the current configuration:
\begin{subequations}\label{eq:def:objectivity}
\begin{align}
\Psi(\boldsymbol{F}) &= \Psi(\boldsymbol{Q} \cdot \boldsymbol{F}) 
\quad \forall \, \boldsymbol{Q} \in \SO(3), \label{eq:def:objectivity:a} \\
\intertext{Consequently, $\Psi$ depends implicitly on the right Cauchy--Green tensor $\boldsymbol{C} = \boldsymbol{F}^\mathsf{T} \cdot \boldsymbol{F}$. This leads to}
\boldsymbol{P} \cdot \boldsymbol{F}^\mathsf{T} 
&= 2 \boldsymbol{F} \cdot \frac{\partial \Psi}{\partial \boldsymbol{C}} \cdot \boldsymbol{F}^\mathsf{T} 
= \boldsymbol{F} \cdot \boldsymbol{P}^\mathsf{T}, \label{eq:def:objectivity:b}
\end{align}
\end{subequations}
ensuring balance of angular momentum.

\item \textbf{Isotropy:}  
The energy must be invariant under rotations of the reference configuration:
\begin{align}
\Psi(\boldsymbol{F}) = \Psi(\boldsymbol{F} \cdot \boldsymbol{Q}^\mathsf{T}) \quad \forall \, \boldsymbol{Q} \in \SO(3).
\label{eq:def:isotropic}
\end{align}

\item \textbf{Growth Conditions:}  
To prevent material self-penetration and ensure physically meaningful behavior, the energy must approach infinity under extreme deformations (also related to coerciveness):
\begin{subequations}
\begin{align}
&\Psi(\boldsymbol{F}) \to \infty \textnormal{ as } \det \boldsymbol{F} \to 0^+, \quad \Psi(\boldsymbol{F}) = \infty \textnormal{ if } \det \boldsymbol{F} \leq 0, \label{eq:growthcond_a}\\
&\Psi(\boldsymbol{F}) \to \infty \textnormal{ as } \{\|\boldsymbol{F}\| + \|\mathrm{cof}\,\boldsymbol{F}\| + \det \boldsymbol{F}\} \to \infty.
\label{eq:growthcond_b}
\end{align}
\end{subequations}
\end{enumerate}
Principles such as determinism, equipresence, and locality are automatically satisfied in this context. Additional constraints, e.g., energy and stress normalization, may be imposed if required \citep{geuken2025}.

From a mathematical perspective, an important property is the lower semicontinuity of the energy functional
\begin{equation*}
\boldsymbol{u} \mapsto I(\boldsymbol{u}) = \int_\Omega \Psi(\nabla \boldsymbol{u}) \, \mathrm{d}V,
\end{equation*}
which guarantees the existence of a minimizer in the set of admissible displacements $\boldsymbol{u}$. Foundational work by Morrey \cite{Mor52, Mor66} and later extensions by Meyers \cite{Mey65}, Acerbi and Fusco \cite{AF84} and Marcellini \cite{Mar85} established that $I$ is lower semicontinuous if $\Psi$ is quasiconvex and satisfies suitable growth and coercivity conditions (see \cite[Theorem 8.29]{Dac08}). However, these conditions exclude energies that take the value $\infty$, making them incompatible with \eqref{eq:growthcond_a} and \eqref{eq:growthcond_b}. To overcome this, Ball \cite{Bal76, ball1977} introduced polyconvexity as a weaker convexity notion that still guarantees lower semicontinuity and existence of minimizers under milder assumptions that are in particular compatible with the growth conditions \eqref{eq:growthcond_a} and \eqref{eq:growthcond_b} (see also \cite[Theorem 8.31]{Dac08}). While polyconvexity implies quasiconvexity and rank-one convexity, polyconvexity is easier to verify and implement. Moreover, it accommodates all physical constraints, making it particularly suitable for modeling elasticity.

\begin{itemize}
\item[7.] \textbf{Polyconvexity:}  
The energy function $\Psi \colon \mathbb{R}^{3\times3} \to \mathbb{R}_\infty \coloneqq \mathbb{R} \cup \{\infty\}$ should be polyconvex, meaning that there exists a convex and lower semicontinuous function $\widehat{\Psi} \colon \mathbb{R}^{3\times3} \times \mathbb{R}^{3\times3} \times \mathbb{R} \to \mathbb{R}_\infty$ such that
\begin{equation}\label{eq:def:Polyconvex}
\Psi(\boldsymbol{F}) = \widehat{\Psi}(\boldsymbol{F}, \mathrm{cof}\,\boldsymbol{F}, \det \boldsymbol{F}) \textnormal{ convex in } \boldsymbol{F}, \mathrm{cof}\,\boldsymbol{F} \textnormal{ and}\, \det \boldsymbol{F} \textnormal{ } \forall \, \boldsymbol{F} \in \mathbb{R}^{3\times3}.
\end{equation}
\end{itemize}

Another relevant material stability condition is the recently rediscovered true-stress–true-strain monotonicity \citep{wollner2026}. While a comparison of these two conditions and their implications for neural network material models would be highly interesting, such an analysis lies beyond the scope of the present work. 

\subsection{The incompressible case}\label{ssec:Incomp}
Although most of the requirements for compressible materials discussed above also apply to incompressible materials, certain modifications have to be introduced to account for the incompressibility constraint
\begin{align}
J = \det \boldsymbol{F} = 1. \label{eq_inc_constr}
\end{align}
First of all, the energy potential 
\begin{align}
\Psi^{\textnormal{inc}}(\bsym F) \quad \forall \textnormal{ } \bsym F \in \R^{3 \times 3} \textnormal{ with } J = \det \boldsymbol{F} = 1
\end{align}
needs to be defined only for $\bsym F$ satisfying Eq.~\eqref{eq_inc_constr}.
In line with \citep{ball1976}, the first Piola--Kirchhoff stress is extended to
\begin{align}
\bsym P = \dfrac{\partial \Psi^{\textnormal{inc}}(\bsym F)}{\partial \bsym F} - p \, \bsym F^{-T} \quad \forall \textnormal{ } \bsym F \in \R^{3 \times 3} \textnormal{ with } J = \det \boldsymbol{F} = 1, \label{eq_inc_stress_split}
\end{align}
wherein $p$ accounts for the surrounding, hydrostatic pressure of the system. The factor $p$ can be interpreted as a Lagrange multiplier for the incompressibility constraint. Often, the deformation gradient is decomposed into a volumetric and an isochoric part (Flory split)
\begin{align}
\bsym F = \bsym F^{\textnormal{vol}} \cdot \bsym{F}^{\textnormal{iso}} \quad \textnormal{with } \bsym F^{\textnormal{vol}} = J^{1/3} \bsym I \quad \textnormal{and } \bsym{F}^{\textnormal{iso}} = J^{-1/3} \bsym F \Rightarrow \det \bsym{F}^{\textnormal{iso}} = 1
\end{align}  
to formulate $\Psi^{\textnormal{inc}}$ only in terms of $\bsym{F}^{\textnormal{iso}}$ and thus directly account for the incompressibility constraint \citep{flory1961}. In this work, this split is omitted and $\Psi^{\textnormal{inc}}$ is evaluated only for deformation states $\bsym F = \bsym{F}^{\textnormal{iso}}$ ($J = \det \bsym F = 1$) as, e.g., in \citep{ball1976}. This holds for all sections and examples concerning incompressibility. 

To apply the concept of polyconvexity, $\Psi^{\textnormal{inc}}$ is identified with its extension by $\infty$, which is given by $\Psi^{\textnormal{inc}}_\infty \colon \R^{3\times 3} \to \R_\infty$ as
\begin{align}\label{eq:FullextensionInfty}
	\Psi^{\textnormal{inc}}_\infty(\bsym F) \coloneq \begin{cases} \Psi^{\textnormal{inc}}(\bsym F) &\textnormal{ if } J = \det \boldsymbol{F} = 1\,,\\
		\infty &\textnormal{ if } J = \det \boldsymbol{F} \neq 1\,.
	\end{cases}
\end{align}
For the existence of a minimizer of the energy functional $I$ in the incompressible case, one is interested in $\Psi^{\textnormal{inc}}_\infty$ to be polyconvex. This motivates the following definition:
\begin{center}
\textit{A function $\Psi^{\textnormal{inc}}$ defined on the subset $\{\bsym F \in \R^{3 \times }\mid\det(\bsym F) =1\}$ is called polyconvex if the corresponding function $\Psi^{\textnormal{inc}}_\infty$ given by Eq.~\eqref{eq:FullextensionInfty} is polyconvex.}
\end{center}
This extension resembles the infinite stiffness of incompressible materials against volumetric deformations. The practical relevance of this extension is two-fold. First, it allows to apply polyconvexity criteria to the incompressible case with minor modifications of the criteria themselves. Second, it is an important prerequisite to prove the necessity of the polyconvexity criterion based on signed singular values in Section \ref{sec_ssv_incompressible}.

Note that $\Psi^{\textnormal{inc}}$ is polyconvex if and only if there exists a convex and lower semicontinuous function $\widehat{\Psi}\colon \R^{3 \times 3} \times \R^{3 \times 3} \to \R_\infty$ such that
\begin{align}
\Psi^{\textnormal{inc}}(\bsym F) = \widehat{\Psi}(\bsym F, \textnormal{cof}\,\bsym F) \textnormal{ convex in } \bsym F \textnormal{ and } \textnormal{cof}\,\bsym F \quad \forall \textnormal{ } \bsym F \in \R^{3 \times 3} \textnormal{ with } J = \det \boldsymbol{F} = 1. \label{eq:psi_inc_polyconvexity}
\end{align}
In Section \ref{sec_ssv_incompressible}, an analogous result is shown for isotropic functions in a dimension-reduced setting. The derivation of Eq.~\ref{eq:psi_inc_polyconvexity} follows analogously.

\subsection{Energy parametrizations in principal stretches and signed singular values}
Reconciling the physical and mathematical constraints outlined above poses a significant challenge -- particularly when aiming to combine frame-indifference, isotropy and polyconvexity in an exact manner. To lay the groundwork for the criteria introduced in the next section, it is first described how frame-indifference and isotropy can be implemented using different parametrizations.

Frame-indifference can be guaranteed a priori by formulating the energy in terms of different arguments, e.g., in terms of eigenvalues of $\bsym C$, principal stretches $\lambda_i$ or signed singular values $\nu_i$ of $\bsym F$. Starting with the eigenvalues of $\bsym C$, the principal stretches  are obtained as their square roots, $\lambda_i = \sqrt{\operatorname{eig}(\bsym C)_i}$. The principal stretches are moreover equal to the singular values of the deformation gradient $\bsym F$. Eventually, the principal stretches (or singular values) are the absolute values of the signed singular values of $\bsym F$, $\lambda_i = |\nu_i|$. Note that the signed singular values $\nu_i$ contain most (physical) information, which is lost for the other arguments by transformation into absolute or quadratic values.

In contrast to the stretches, the signed singular values can become negative by rotation or -- to be more precise -- because the decomposition 
\begin{align}
\bsym F = \bsym{R}_1 \cdot \bsym{\textbf{diag}}(\bsym \nu) \cdot \bsym R_2 \quad \textnormal{with }\boldsymbol{\nu} = (\nu_1, \nu_2, \nu_3)
\end{align}
of $\bsym F$ into two rotations $\bsym R_1$, $\bsym R_2 \in \SO(3)$ and a diagonal matrix $\textbf{diag}(\bsym \nu)$ is only unique up to permutations and an even number of sign changes for the entries $\nu_1, \nu_2, \nu_3$ as also motivated by Fig.~\ref{fig_ssv_rotation}. However, the choice of their sign does not affect the determinant (orientation preserving), yielding $\nu_1 \, \nu_2 \, \nu_3 = \det \boldsymbol{F} = J > 0$. Signed singular values will be of particular importance for the following investigation. A key property for this is:
\begin{center}
\textit{For every frame-indifferent and isotropic function $\Psi(\bsym F)$ there exists a unique function $\widetilde{\Psi}(\nu_1, \nu_2, \nu_3)$ that characterizes $\Psi$ in terms of the signed singular values of $\bsym F$, cf.~\citep{wiedemann2026}.}
\end{center}
This implies that $\widetilde{\Psi}$ is $\Pi_3$-invariant with 
\begin{align}
\Pi_3 \coloneq \lbrace \boldsymbol B \cdot \textbf{diag}(\bsym \epsilon) \, \vert \, \boldsymbol B \in \mathrm{Perm}(3), \, \bsym \epsilon \in \lbrace -1,1 \rbrace^3, \, \epsilon_1 \, \epsilon_2 \, \epsilon_3 = 1 \rbrace,
\end{align} 
where Perm$(3)$ denotes the set of permutation matrices \citep{wiedemann2026}. In simpler terms, $\Pi_3$-invariance incorporates the four symmetries
\begin{align}
\widetilde{\Psi} (\nu_1, \nu_2, \nu_3) = \widetilde{\Psi} (-\nu_1, -\nu_2, \nu_3) = \widetilde{\Psi} (-\nu_1, \nu_2, -\nu_3) = \widetilde{\Psi} (\nu_1, -\nu_2, -\nu_3)
\end{align}
along with the six invariances of $\widetilde{\Psi}$ with respect to all permutations of the signed singular values, e.g., $\widetilde{\Psi} (\nu_1, \nu_2, \nu_3) = \widetilde{\Psi} (\nu_2, \nu_3, \nu_1)$ (see Tab.~\ref{tab:input_permutations} for all combinations). The condition $\epsilon_1 \, \epsilon_2 \, \epsilon_3 = 1$ enforces the local deformation (deformation gradient) to be orientation preserving. According to \citep[Lemma 2.4]{wiedemann2026} the statement above is even sharper, namely that the set of frame-indifferent and isotropic functions $\Psi(\bsym F)$ can be identified with the set of $\Pi_3$-invariant functions $\widetilde{\Psi}(\nu_1, \nu_2, \nu_3)$ that characterize $\Psi$ in terms of the signed singular values of $\bsym F$.
\begin{center}
\textit{Analogously, for every frame-indifferent and isotropic function $\Psi(\bsym F)$ there exists a unique function $\breve\Psi(\lambda_1, \lambda_2, \lambda_3)$ that characterizes $\Psi$ in terms of the principal stretches (singular values) of $\bsym F$, cf.~\citep{wiedemann2026}.} 
\end{center}
In that case, this implies $\breve\Psi$ to be invariant with respect to permutations given by Perm$(3)$. The symmetry regarding the sign is already included by the restricted domain $\boldsymbol{\lambda} = (\lambda_1, \lambda_2, \lambda_3) \in (0,\infty)^3$. Clearly, this parametrization includes the classical formulation in matrix invariants of $\bsym C$, which directly ensures frame-indifference and isotropy.

\begin{figure}[ht!]
\centering
\vspace{0.7em}
\includegraphics[width=0.5\textwidth]{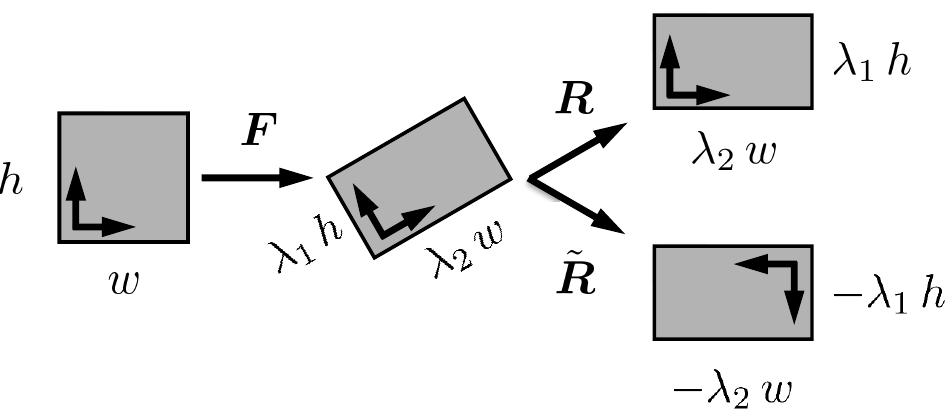}
\caption{Non-uniqueness of the sign of the singular values motivated by the deformation of a block.} \label{fig_ssv_rotation}
\end{figure}

\section{Sufficient and necessary criteria for frame-indifferent, isotropic polyconvex energies of compressible and incompressible materials}
Within this section, a sufficient and necessary criterion and different sufficient criteria for frame-indifferent, isotropic polyconvex energies are presented. Each of them serves as the basis of a neural network framework developed in the following section. A summary of all investigated criteria is provided in Fig.~\ref{fig_crit_overview}. The first two of them are based on signed singular values. While the first one provides a sufficient and necessary criterion, the second one is a reduced formulation in terms of the elementary symmetric polynomials of the signed singular values. The third one is an improved version of the well-known sufficient criterion of \citet{ball1976} based on principal stretches. Finally, a reduced formulation in terms of matrix invariants is provided. These criteria represent the state of the art for both -- polyconvexity itself as well as neural networks in constitutive modeling -- and thus allow a comprehensive analysis.
Each criterion is, first, given for the general, compressible case and then the associated criterion for the incompressible case is directly derived from that by removing the determinant as an argument of the convex representative of the energy density. This simple adjustment is possible for every criterion under investigation since their energy parametrizations allow a clear distinction regarding convexity in the deformation gradient $\bsym F$, its cofactor $\textnormal{cof}\,\boldsymbol{F}$ and its determinant $J = \det \boldsymbol{F}$.
 
\begin{figure}[ht!]
\centering
\includegraphics[width=\textwidth]{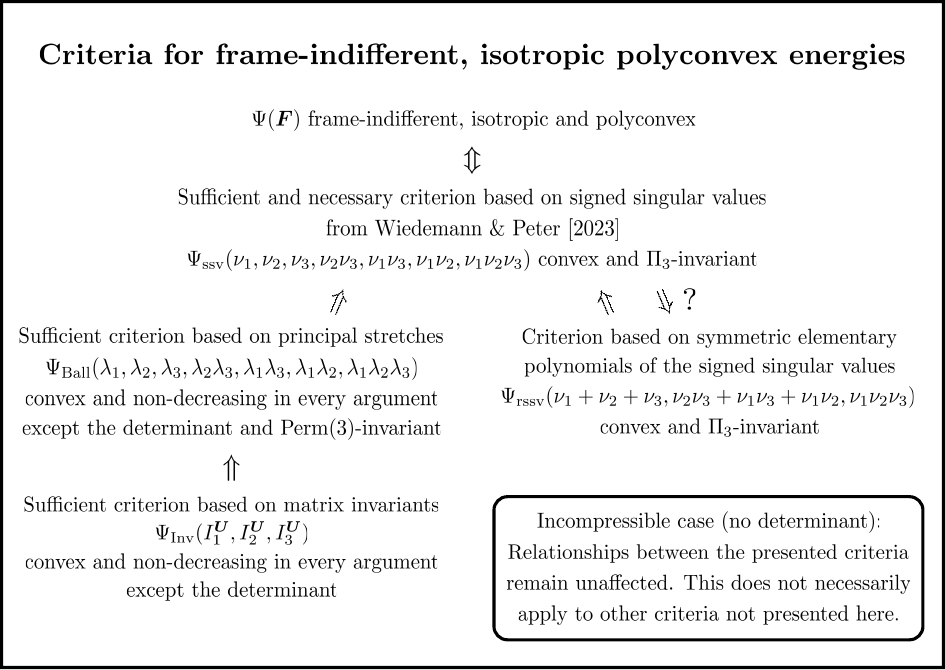}
\caption{Overview and relations between the investigated criteria for frame-indifferent, isotropic polyconvex energies of compressible materials. For incompressibility the determinant's influence via arguments $\nu_1 \nu_2 \nu_3, \lambda_1 \lambda_2 \lambda_3$ and $ I_3^{\boldsymbol U}$ is omitted. The open, ?-marked relation is indicated by the numerical results in Section \ref{sec_examples}.} \label{fig_crit_overview}
\end{figure} 
 
\subsection{Sufficient and necessary criterion based on signed singular values}
\subsubsection{The compressible case}
Following \citep{wiedemann2026,geuken2025}, $\widetilde{\Psi} (\nu_1, \nu_2, \nu_3)$ is called singular value polyconvex if the associated energy $\Psi$ is polyconvex, i.e.,~
\begin{equation}
\widetilde{\Psi} \textnormal{ singular value polyconvex} :\Leftrightarrow \Psi \textnormal{ polyconvex}.
\end{equation}
Then, Theorem 1.4 in \citep{wiedemann2026} provides a sufficient and necessary criterion for frame-indifferent, isotropic polyconvex energies as follows:
\begin{theorem}[\textbf{Singular value polyconvexity, cf.~\citep{wiedemann2023, wiedemann2026}}]\label{theo_singvalpolyconv}
A $\Pi_3$-invariant function $\widetilde{\Psi}\colon \, \R^3 \rightarrow \R_\infty$ is singular value polyconvex (and thus $\Psi$ is polyconvex) if and only if there exists a convex and lower semicontinuous function $\Psi_{\textnormal{ssv}} \colon \R^7 \rightarrow \R_\infty$ such that
\begin{align}
\widetilde{\Psi} (\nu_1, \nu_2, \nu_3) = \Psi_{\textnormal{ssv}} (\underbrace{\nu_1, \nu_2, \nu_3}_{\bsym F}, \underbrace{\nu_2 \nu_3, \nu_1 \nu_3, \nu_1 \nu_2}_{\textnormal{cof}\,\bsym F}, \underbrace{\nu_1 \nu_2 \nu_3}_{\textnormal{det}\,\bsym F}) \quad \forall \, \boldsymbol \nu \in \R^3 \text{ with }\boldsymbol{\nu} = (\nu_1, \nu_2, \nu_3). \label{eq_psiisoandconv}
\end{align}
\end{theorem}
The proof of Theorem \ref{theo_singvalpolyconv} was given by \citet{wiedemann2023, wiedemann2026}.
The first three arguments of $\Psi_{\textnormal{ssv}}$ are connected to the convexity with respect to $\bsym F$, the next three to the convexity with respect to $\textnormal{cof}\,\bsym F$ and the last argument is the determinant $\textnormal{det}\,\bsym F$. Using
\begin{equation}
m(\boldsymbol{\nu}) \coloneqq ( \nu_1, \nu_2, \nu_3, \nu_2 \nu_3, \nu_1 \nu_3, \nu_1 \nu_2, \nu_1 \nu_2 \nu_3)
\end{equation}
one can rewrite Eq.~\eqref{eq_psiisoandconv} in the more compact form
\begin{align}
\widetilde{\Psi} = \Psi_{\textnormal{ssv}}\circ m.
\end{align}

This criterion was already used in \citep{geuken2025} in conjunction with neural networks and has proven to be very powerful for constitutive modeling and the approximation of polyconvex hulls, see also \citep{neumeier2024, balazi2025} for further numerical approximations of polyconvex hulls. Even further, the associated networks can approximate any frame-indifferent, isotropic polyconvex energy up to arbitrary precision.

\subsubsection{The incompressible case} \label{sec_ssv_incompressible}
In this section, incompressible energies $\Psi^{\textnormal{inc}}$ are considered, as described in Section \ref{ssec:Incomp}. With \citep[Lemma 2.4]{wiedemann2026} one can identify incompressible functions $\Psi^{\textnormal{inc}}$ that are defined on $\{\bsym F \in \R^{3 \times 3} \mid J = \det(\bsym F) =1\}$ with $\Pi_3$-invariant functions $\widetilde{\Psi}^{\textnormal{inc}}$ defined on $\{\bsym \nu \in \R^3 \mid \nu_1 \nu_2 \nu_3 = 1\}$. 
In line with the compressible case, a function $\widetilde{\Psi}^{\textnormal{inc}}$ defined on $\bsym \nu\in \R^3$ with $\nu_1 \nu_2 \nu_3= 1$ is called singular value polyconvex if the corresponding function $\Psi^{\textnormal{inc}}$ is polyconvex. Recall that the polyconvexity of $\Psi^{\textnormal{inc}}$ is defined by the polyconvexity of $\Psi^{\textnormal{inc}}_\infty$. Thus, $\widetilde{\Psi}^{\textnormal{inc}}$ is singular value polyconvex if and only if its extension by $\infty$, which is given by $\widetilde{\Psi}^{\textnormal{inc}}_\infty \colon \R^3 \to \R_\infty$
\begin{equation}\label{eq:Psiinfty=Psi}
 \widetilde{\Psi}^{\textnormal{inc}}_\infty(\bsym \nu) \coloneq \begin{cases} \widetilde{\Psi}^{\textnormal{inc}}(\bsym \nu) &\textnormal{ if } \nu_1 \nu_2 \nu_3 = 1\,,\\
 \infty &\textnormal{ if } \nu_1 \nu_2 \nu_3 \neq 1\,,
 \end{cases}
\end{equation}
is singular value polyconvex.
The following corollary is a direct consequence of Theorem \ref{theo_singvalpolyconv}.
\begin{corollary}[\textbf{Singular value polyconvexity for incompressibility}]\label{theo_singvalpolyconv_incompr}
A $\Pi_3$-invariant function $\widetilde{\Psi}^{\textnormal{inc}}\colon \{\bsym \nu \in \R^3 \mid \nu_1 \nu_2 \nu_3 = 1\} \rightarrow \R_\infty$ is singular value polyconvex (and thus $\Psi^{\textnormal{inc}}$ is polyconvex) under the incompressibility constraint \eqref{eq_inc_constr} if and only if there exists a convex and lower semicontinuous function $\Psi_{\textnormal{ssv}}^{\textnormal{inc}}\colon \R^6 \rightarrow \R_\infty$ such that
\begin{align}
\widetilde{\Psi}^{\textnormal{inc}} (\nu_1, \nu_2, \nu_3) = \Psi_{\textnormal{ssv}}^{\textnormal{inc}} (\nu_1, \nu_2, \nu_3, \nu_2 \nu_3, \nu_1 \nu_3, \nu_1 \nu_2) \quad \forall \, \boldsymbol \nu \in \R^3 \text{ with }\boldsymbol{\nu} = (\nu_1, \nu_2, \nu_3) \text{ and } \nu_1 \nu_2 \nu_3 = 1. \label{eq_psiisoandconv_incompr}
\end{align}
\end{corollary}
Using
\begin{equation}
m^{\textnormal{inc}}(\boldsymbol{\nu}) \coloneqq ( \nu_1, \nu_2, \nu_3, \nu_2 \nu_3, \nu_1 \nu_3, \nu_1 \nu_2)
\label{eq:elementary_polynomials}
\end{equation}
one can rewrite Eq.~\eqref{eq_psiisoandconv_incompr} in the more compact form
\begin{align}
\widetilde{\Psi}^{\textnormal{inc}} = \Psi_{\textnormal{ssv}}^{\textnormal{inc}} \circ m^{\textnormal{inc}}.
\end{align}
\begin{proof}
Let $\widetilde{\Psi}^{\textnormal{inc}}$ be $\Pi_3$-invariant and singular value polyconvex, then the function $\Psi_{\textnormal{ssv}}^{\textnormal{inc}}$ for \eqref{eq_psiisoandconv_incompr} can be constructed with Theorem \ref{theo_singvalpolyconv}: Since $\widetilde{\Psi}^{\textnormal{inc}}$ is singular value polyconvex, there exists a convex and lower semicontinuous function $\Psi_{\textnormal{ssv}} \colon \R^7 \rightarrow \R_\infty$ that satisfies Eq.~\eqref{eq_psiisoandconv} for $\widetilde{\Psi}= \widetilde{\Psi}^{\textnormal{inc}}_\infty$. The function $\Psi_{\textnormal{ssv}}^{\textnormal{inc}} \colon \R^6\to \R_\infty$ defined by $\Psi_{\textnormal{ssv}}^{\textnormal{inc}}(\bsym x) \coloneqq \Psi_{\textnormal{ssv}}(\bsym x,1)$ for $\bsym x \in \R^6$ is convex and lower semicontinuous as it is the restriction of a convex and lower semicontinuous function. It satisfies Eq.~\eqref{eq_psiisoandconv_incompr} for $\boldsymbol{\nu} = (\nu_1, \nu_2, \nu_3)$ with $\nu_1 \nu_2 \nu_3 = 1$:
\begin{equation}
\widetilde{\Psi}^{\textnormal{inc}} (\boldsymbol{\nu}) \overset{\eqref{eq:Psiinfty=Psi}}{=} \widetilde{\Psi}^{\textnormal{inc}}_\infty (\boldsymbol{\nu}) \overset{\textnormal{Theorem } \ref{theo_singvalpolyconv}}{=} \Psi_{\textnormal{ssv}} (m(\boldsymbol{\nu})) 
=
\Psi_{\textnormal{ssv}} (m^{\textnormal{inc}}(\boldsymbol{\nu}),1) 
=
\Psi_{\textnormal{ssv}}^{\textnormal{inc}} (m^{\textnormal{inc}}(\boldsymbol{\nu})) \,.
\end{equation}

For the other direction, let $\Psi_{\textnormal{ssv}}^{\textnormal{inc}} \colon \R^6\to \R_\infty$
be a convex and lower semicontinuous function that satisfies Eq.~\eqref{eq_psiisoandconv_incompr}. 
The idea is to construct a convex and lower semicontinuous function $\Psi_{\textnormal{ssv}} \colon \R^7 \rightarrow \R_\infty$ that satisfies Eq.~\eqref{eq_psiisoandconv} for
$\widetilde{\Psi}= \widetilde{\Psi}^{\textnormal{inc}}_\infty$. Then, the singular value polyconvexity of $\widetilde{\Psi}^{\textnormal{inc}}$ follows from Theorem \ref{theo_singvalpolyconv}.
Define $\Psi_{\textnormal{ssv}}$ by $\Psi_{\textnormal{ssv}}(\bsym x) \coloneqq 
\Psi_{\textnormal{ssv}}^{\textnormal{inc}}(\bsym x)$ for $x_7 = 1$ and $\Psi_{\textnormal{ssv}}(\bsym x) \coloneqq \infty$ for $x_7 \neq 1$. This function is the extension of a convex function outside the convex domain $\{\bsym x \in \R^7 \mid x_7 = 1 \}$ by $\infty$ and, thus, it is convex.
Since $\Psi_{\textnormal{ssv}}^{\textnormal{inc}}$ is lower semicontinuous and the set $\{\bsym x \in \R^7 \mid x_7 = 1 \}$ is closed, the extension by $\infty$ is lower semicontinuous. It remains to check Eq.~\eqref{eq_psiisoandconv} for $\widetilde{\Psi}= \widetilde{\Psi}^{\textnormal{inc}}_\infty$. It holds, since:
\begin{equation*}
 \widetilde{\Psi}^{\textnormal{inc}}_\infty (\boldsymbol{\nu}) 
 \overset{\eqref{eq:Psiinfty=Psi}}{=}
 \begin{cases}
\widetilde{\Psi}^{\textnormal{inc}}(\boldsymbol{\nu}) &\text{if } \nu_1\nu_2\nu_3=1\,,
\\
\infty&\text{else}
 \end{cases}
 \overset{\eqref{eq_psiisoandconv_incompr}}{=}
 \begin{cases}
\Psi^{\textnormal{inc}}_{\textnormal{ssv}}(m^{\textnormal{inc}}(\boldsymbol{\nu})) &\text{if } m(\boldsymbol{\nu})_7=1 \,,
\\
\infty &\text{else}
 \end{cases}
 =\Psi_{\textnormal{ssv}}(m(\boldsymbol{\nu})).
\end{equation*}
\end{proof}

As will be shown later, neural network models based on this criterion can approximate any frame-indifferent, isotropic polyconvex energy under the incompressibility constraint.

\subsection{Reduced criterion based on elementary symmetric polynomials of the signed singular values}
\subsubsection{The compressible case}
\citet{wiedemann2026} additionally derived a new criterion in terms of the elementary symmetric polynomials of the signed singular values ($\nu_1 + \nu_2 + \nu_3, \nu_2 \nu_3 + \nu_1 \nu_3 + \nu_1 \nu_2, \nu_1 \nu_2 \nu_3$). Therein, a $\Pi_3$-invariant function $\widetilde{\Psi}$ is singular value polyconvex if it is given by
\begin{equation}
\widetilde{\Psi} (\nu_1, \nu_2, \nu_3) = \Psi_{\textnormal{rssv}} (\nu_1 + \nu_2 + \nu_3, \nu_2 \nu_3 + \nu_1 \nu_3 + \nu_1 \nu_2, \nu_1 \nu_2 \nu_3) \quad \forall \, \boldsymbol \nu \in \R^3 \textnormal{ with }\boldsymbol{\nu} = (\nu_1, \nu_2, \nu_3) \label{eq_redpolycrit}
\end{equation}
with $\Psi_{\textnormal{rssv}}$ being convex. In this setting, the Perm$(3)$-invariance is a priori fulfilled, however, the symmetries regarding the sign have to be accounted for when formulating $\Psi_{\textnormal{rssv}}$, cf. Tab.~\ref{tab:input_permutations_rssv}. It is not yet known whether this criterion is also necessary or which frame-indifferent, isotropic polyconvex energy functions may be unattainable by this formulation.

\subsubsection{The incompressible case}
The corresponding criterion for the incompressible case is again obtained by simply dropping the determinant. In that case, a $\Pi_3$-invariant function $\widetilde{\Psi}^{\textnormal{inc}}$ is singular value polyconvex (and thus $\Psi^{\textnormal{inc}}$ is polyconvex) under the incompressibility constraint if it is given by
\begin{equation}
\widetilde{\Psi}^{\textnormal{inc}} (\nu_1, \nu_2, \nu_3) = \Psi_{\textnormal{rssv}}^{\textnormal{inc}} (\nu_1 + \nu_2 + \nu_3, \nu_2 \nu_3 + \nu_1 \nu_3 + \nu_1 \nu_2) \quad \forall \, \boldsymbol \nu \in \R^3 \textnormal{ with }\boldsymbol{\nu} = (\nu_1, \nu_2, \nu_3) \textnormal{ and } \nu_1 \nu_2 \nu_3 = 1 \label{eq_redpolycrit_incompr}
\end{equation}
with $\Psi_{\textnormal{rssv}}^{\textnormal{inc}}$ being convex. As for the general case, the Perm$(3)$-invariance is a priori fulfilled within this setting. However, the symmetries regarding the sign have to be accounted for when formulating $\Psi_{\textnormal{rssv}}^{\textnormal{inc}}$, cf. Tab.~\ref{tab:input_permutations_rssv}.

\subsection{Sufficient criterion based on principal stretches}
\subsubsection{The compressible case}
A new, sufficient criterion for polyconvexity formulated in principal stretches $\lambda_i$ can be derived as a special case of Theorem~\ref{theo_singvalpolyconv} through the relationship $\lambda_i = \vert \nu_i \vert$.
\begin{theorem}[\textbf{Sufficient polyconvexity criterion based on principal stretches}]\label{theo_princstretchpolyconv}
Let $\Psi \colon \R^{3 \times 3} \to \R_\infty$ be given by
\begin{align}\label{eq:6785}
\Psi(\bsym F) =
\begin{cases}
\breve\Psi(\lambda_1, \lambda_2, \lambda_3) \quad \forall \, \boldsymbol \lambda \in \R_+^3 \text{ with }\boldsymbol{\lambda} = (\lambda_1, \lambda_2, \lambda_3) &\text{if } J = \det \bsym F \geq 0\,,
\\
\infty &\text{if } J= \det \bsym F < 0
\end{cases}
\end{align}
where $\lambda_i$ are the principal stretches of $\bsym F$
and $\breve\Psi \colon [0, \infty)^3\to \R_\infty$ is $\operatorname{Perm}(3)$-invariant and given by
\begin{equation}
\breve\Psi(\lambda_1, \lambda_2, \lambda_3) = \Psi_{\textnormal{Ball}}(\lambda_1, \lambda_2, \lambda_3, \lambda_2 \lambda_3, \lambda_1 \lambda_3, \lambda_1 \lambda_2, \lambda_1 \lambda_2 \lambda_3)
\end{equation}
for a convex and lower semicontinuous function
$\Psi_{\textnormal{Ball}}\colon \R^6 \times [0, \infty) \to \R_\infty$ that is non-decreasing in every argument except the last argument ($J = \det \boldsymbol{F} = \lambda_1 \lambda_2 \lambda_3$).
Then, $\Psi$ is polyconvex.
\end{theorem}
The proof is provided in Appendix \ref{proof_ball}. This criterion bears similarities to Ball’s well-known criterion for polyconvexity. In recognition of Ball’s work, the energy function is therefore denoted by $\Psi_{\textnormal{Ball}}$. However, Ball's original criterion requires a stricter permutation invariance: 
\begin{theorem}[\textbf{Ball's sufficient polyconvexity criterion based on principal stretches \citep{ball1976}}] \label{theo_ball_original}
Consider energy functions $\Psi(\bsym F)$ defined on sets of the form $U = \{ \bsym F \in M^{3 \times 3}\colon \det \bsym F \in K \}$, where $K\subseteq \R_+$ is convex. Let
\begin{equation}
\Psi(\bsym F) = \psi(\lambda_1, \lambda_2, \lambda_3, \lambda_2 \lambda_3, \lambda_1 \lambda_3, \lambda_1 \lambda_2, \lambda_1 \lambda_2 \lambda_3),
\end{equation}
where $\lambda_i$ are the singular values (principal stretches) of $\bsym F \in U$ and where $\psi\colon \R_+^6 \times K \to \R$ is convex and satisfies
\begin{align}
&\textnormal{(a) } \psi(\bsym B_1 \cdot \bsym x, \bsym B_2 \cdot \bsym{y}, \delta) \quad \forall \, \bsym B_1, \bsym B_2 \in \operatorname{Perm}(3) \text{ and all } \bsym x, \bsym y \in R_+^3, \, \delta\in K,\\
&\textnormal{(b) } \psi(x_1, x_2, x_3, y_1, y_2, y_3, \delta) \text{ is non-decreasing in each } x_i, y_j,
\end{align}
wherein $\bsym x = (x_1,x_2,x_3)$ and $\bsym y = (y_1,y_2,y_3)$ denote the first three and the fourth through sixth arguments of $\psi$, respectively. Then $\Psi$ is polyconvex on $U$.
\end{theorem}
The improvement of Theorem \ref{theo_princstretchpolyconv} over Ball's Theorem \ref{theo_ball_original} is the more relaxed permutation invariance
\begin{align}
\Psi_{\textnormal{Ball}}(\bsym B \cdot \bsym \lambda, \bsym B \cdot \bsym{\check{\lambda}}, \lambda_1 \lambda_2 \lambda_3) = \Psi_{\textnormal{Ball}}(\bsym \lambda, \bsym{\check{\lambda}}, \lambda_1 \lambda_2 \lambda_3) \quad \forall \, \bsym B \in \textnormal{Perm}(3) &\text{ with }\boldsymbol{\lambda} = (\lambda_1, \lambda_2, \lambda_3)\\ &\text{ and } \bsym{\check{\lambda}} = (\lambda_2 \lambda_3, \lambda_1 \lambda_3, \lambda_1 \lambda_2) \nonumber
\end{align}
since $\bsym \lambda$ and $\bsym{\check{\lambda}}$ as well as their permutations are connected, which can be derived from Theorem~\ref{theo_singvalpolyconv}. Nevertheless, this criterion is still only sufficient but not necessary and rules out some frame-indifferent, isotropic polyconvex functions as numerical examples in Section \ref{sec_mielke_energy} will show.
\subsubsection{The incompressible case}
The derivation of the above criterion for the incompressible case is also straightforward and reads:
\begin{corollary}[\textbf{Sufficient polyconvexity criterion for incompressibility based on principal stretches}]\label{theo_princstretchpolyconv_incompr}
Let $\Psi^{\mathrm{inc}}$ be given by
\begin{align}
\Psi^{\textnormal{inc}}(\bsym F) = \breve\Psi^{\textnormal{inc}}(\lambda_1, \lambda_2, \lambda_3) = \Psi_{\textnormal{Ball}}^{\textnormal{inc}}(\lambda_1, \lambda_2, \lambda_3, \lambda_2 \lambda_3, \lambda_1 \lambda_3, \lambda_1 \lambda_2) \quad \forall \, \boldsymbol \lambda \in \R_+^3 &\text{ with }\boldsymbol{\lambda} = (\lambda_1, \lambda_2, \lambda_3) \\ &\text{ and } J = \lambda_1 \lambda_2 \lambda_3 = 1 \nonumber
\end{align}
with $\Psi_{\textnormal{Ball}}^{\textnormal{inc}}$ convex and non-decreasing and chosen such that the associated $\breve\Psi$ is Perm$(3)$-invariant, then $\Psi^{\textnormal{inc}}$ is polyconvex.
\end{corollary}
The proof is analogous to the one of Corollary~\ref{theo_singvalpolyconv_incompr}. It should be emphasized once more that this criterion is sufficient but not necessary as will be shown later.

\subsection{Sufficient criterion based on matrix invariants}
\subsubsection{The compressible case}
A criterion formulated in matrix invariants can be directly motivated from Ball's criterion by using coupled combinations of the principal stretches. Within this work, the matrix invariants ($I_1^{\boldsymbol U} = \lambda_1 + \lambda_2 + \lambda_3$, $I_2^{\boldsymbol U} = \lambda_2 \lambda_3 + \lambda_1 \lambda_3 + \lambda_1 \lambda_2, I_3^{\boldsymbol U} = J = \lambda_1 \lambda_2 \lambda_3$) of the right stretch tensor $\bsym U = \sqrt{\bsym C} = \sqrt{\bsym{F^T} \cdot \bsym F}$ are used; see, for example, \citep{steigmann2003} for more details. Let $\Psi$ be given by
\begin{align}
\Psi(\bsym F) = \Psi_{\textnormal{Inv}}(I_1^{\boldsymbol U}, I_2^{\boldsymbol U}, I_3^{\boldsymbol U}) \quad \forall \, \bsym F \in \R^{3 \times 3} \textnormal{ with } J = \det \bsym F \geq 0
\end{align}
with $\Psi_{\textnormal{Inv}}$ convex and non-decreasing in the first two arguments, then $\Psi$ is polyconvex. In that setting frame-indifference and isotropy of $\Psi$ is automatically ensured by working with invariants which in turn implies Perm$(3)$-invariance of $\Psi_{\textnormal{Inv}}$ if one draws the connection to the other criteria.\\
Certainly, the invariants of the right Cauchy--Green tensor $\bsym C$ could also be used. However, it was shown in \citep{geuken2025} that this criterion is not capable of approximating, e.g., energy functions exhibiting linear growth with respect to the stretch. Such linear energy growth corresponds to stress plateaus, which are observed, e.g., for phase transitions, polymers and biological materials such as the periodontal ligament \citep{rees1997,kurzeja2025}. 

\subsubsection{The incompressible case}
The criterion reduces to the arguments $I_1^{\boldsymbol U} = \lambda_1 + \lambda_2 + \lambda_3$ and $I_2^{\boldsymbol U} = \lambda_2 \lambda_3 + \lambda_1 \lambda_3 + \lambda_1 \lambda_2$ for incompressibility. Let $\Psi^{\textnormal{inc}}$ be given by
\begin{align}
\Psi^{\textnormal{inc}}(\bsym F) = \Psi_{\textnormal{Inv}}^{\textnormal{inc}}(I_1^{\boldsymbol U}, I_2^{\boldsymbol U}) \quad \forall \, \bsym F \in \R^{3 \times 3} \textnormal{ with } J = \det \bsym F = 1
\end{align}
with $\Psi_{\textnormal{Inv}}^{\textnormal{inc}}$ convex and non-decreasing, then $\Psi^{\textnormal{inc}}$ is polyconvex. As before, frame-indifference and isotropy is guaranteed by working with invariants.

\begin{table}[h!]
\caption{Input permutations of the signed singular values for the neural network employed in Eq.~\eqref{nn_energy}.}
\renewcommand{\arraystretch}{2}
\resizebox{\textwidth}{!}{%
\begin{tabular}{|l|l|l|}
\hline
$\boldsymbol x^{(1)} = \left( \nu_1, \nu_2, \nu_3, \nu_2 \nu_3, \nu_1 \nu_3, \nu_1 \nu_2 \right)$ & $\boldsymbol x^{(2)} = \left( -\nu_1, -\nu_2, \nu_3, -\nu_2 \nu_3, -\nu_1 \nu_3, \nu_1 \nu_2 \right)$ & $\boldsymbol x^{(3)} = \left( -\nu_1, \nu_2, -\nu_3, -\nu_2 \nu_3, \nu_1 \nu_3, -\nu_1 \nu_2 \right)$ \\\hline
$\boldsymbol x^{(4)} = \left( \nu_1, -\nu_2, -\nu_3, \nu_2 \nu_3, -\nu_1 \nu_3, -\nu_1 \nu_2 \right)$ & $\boldsymbol x^{(5)} = \left( \nu_1, \nu_3, \nu_2, \nu_2 \nu_3, \nu_1 \nu_2, \nu_1 \nu_3 \right)$ & $\boldsymbol x^{(6)} = \left( -\nu_1, -\nu_3, \nu_2, -\nu_2 \nu_3, -\nu_1 \nu_2, \nu_1 \nu_3 \right)$ \\ \hline
$\boldsymbol x^{(7)} = \left( -\nu_1, \nu_3, -\nu_2, -\nu_2 \nu_3, \nu_1 \nu_2, -\nu_1 \nu_3 \right)$ & $\boldsymbol x^{(8)} = \left( \nu_1, -\nu_3, -\nu_2, \nu_2 \nu_3, -\nu_1 \nu_2, -\nu_1 \nu_3 \right)$ & $\boldsymbol x^{(9)} = \left( \nu_2, \nu_1, \nu_3, \nu_1 \nu_3, \nu_2 \nu_3, \nu_1 \nu_2 \right)$ \\ \hline
$\boldsymbol x^{(10)} = \left( -\nu_2, -\nu_1, \nu_3, -\nu_1 \nu_3, -\nu_2 \nu_3, \nu_1 \nu_2 \right)$ & $\boldsymbol x^{(11)} = \left( -\nu_2, \nu_1, -\nu_3, -\nu_1 \nu_3, \nu_2 \nu_3, -\nu_1 \nu_2 \right)$ & $\boldsymbol x^{(12)} = \left( \nu_2, -\nu_1, -\nu_3, \nu_1 \nu_3, -\nu_2 \nu_3, -\nu_1 \nu_2 \right)$ \\ \hline
$\boldsymbol x^{(13)} = \left( \nu_3, \nu_1, \nu_2, \nu_1 \nu_2, \nu_2 \nu_3, \nu_1 \nu_3 \right)$ & $\boldsymbol x^{(14)} = \left( -\nu_3, -\nu_1, \nu_2, -\nu_1 \nu_2, -\nu_2 \nu_3, \nu_1 \nu_3 \right)$ & $\boldsymbol x^{(15)} = \left( -\nu_3, \nu_1, -\nu_2, -\nu_1 \nu_2, \nu_2 \nu_3, -\nu_1 \nu_3 \right)$ \\ \hline
$\boldsymbol x^{(16)} = \left( \nu_3, -\nu_1, -\nu_2, \nu_1 \nu_2, -\nu_2 \nu_3, -\nu_1 \nu_3 \right)$ & $\boldsymbol x^{(17)} = \left( \nu_2, \nu_3, \nu_1, \nu_1 \nu_3, \nu_1 \nu_2, \nu_2 \nu_3 \right)$ & $\boldsymbol x^{(18)} = \left( -\nu_2, -\nu_3, \nu_1, -\nu_1 \nu_3, -\nu_1 \nu_2, \nu_2 \nu_3 \right)$ \\ \hline
$\boldsymbol x^{(19)} = \left( -\nu_2, \nu_3, -\nu_1, -\nu_1 \nu_3, \nu_1 \nu_2, -\nu_2 \nu_3 \right)$ & $\boldsymbol x^{(20)} = \left( \nu_2, -\nu_3, -\nu_1, \nu_1 \nu_3, -\nu_1 \nu_2, -\nu_2 \nu_3 \right)$ & $\boldsymbol x^{(21)} = \left( \nu_3, \nu_2, \nu_1, \nu_1 \nu_2, \nu_1 \nu_3, \nu_2 \nu_3 \right)$ \\ \hline
$\boldsymbol x^{(22)} = \left( -\nu_3, -\nu_2, \nu_1, -\nu_1 \nu_2, -\nu_1 \nu_3, \nu_2 \nu_3 \right)$ & $\boldsymbol x^{(23)} = \left( -\nu_3, \nu_2, -\nu_1, -\nu_1 \nu_2, \nu_1 \nu_3, -\nu_2 \nu_3 \right)$ & $\boldsymbol x^{(24)} = \left( \nu_3, -\nu_2, -\nu_1, \nu_1 \nu_2, -\nu_1 \nu_3, -\nu_2 \nu_3 \right)$ \\ \hline
\end{tabular}%
}
\label{tab:input_permutations}
\end{table}

\section{Isotropic polyconvex neural networks for incompressibility} \label{sec_iso_poly_nns}
Within this section suitable neural network approximations $\Psi^\nn$ for $\Psi^{\textnormal{inc}}$ are proposed based on the different criteria for polyconvexity from the previous section. To ensure convexity in the input arguments that every criterion requires, all subsequent models are based on input convex neural networks (ICNNs) \citep{amos2017}. ICNNs are defined by
\begin{align}
\boldsymbol z_{i+1} = g_i \left(\boldsymbol W_i^{(\boldsymbol z)} \cdot \boldsymbol z_i + \boldsymbol W_i^{(\boldsymbol x)} \cdot \boldsymbol x + \boldsymbol b_i \right) \textnormal{ for } i = 0,...,k-1 \textnormal{ and } \nn(\boldsymbol x;\theta) = \boldsymbol z_k, \label{eq:nn_architecture}
\end{align}
where $\boldsymbol x$ is the input vector, $\boldsymbol z_i$ are the layer activations ($\boldsymbol z_0 \equiv 0,$ $\boldsymbol W_0^{(\boldsymbol z)} \equiv 0$), $\theta = \lbrace \boldsymbol W_{0:k-1}^{(\boldsymbol x)}, \boldsymbol W_{1:k-1}^{(\boldsymbol z)}, \boldsymbol b_{0:k-1} \rbrace$ are the weights, $g_i$ are non-linear activation functions and $k$ is the number of layers. Given this architecture, $\nn$ is convex in $\boldsymbol x$ if all $\boldsymbol W_{1:k-1}^{(\boldsymbol z)}$ are non-negative and all activation functions $g_i$ are convex and non-decreasing \citep{amos2017}. ICNNs are employed in this work because of their representational power. A universal approximation theorem for convex functions was proven by \citet{huang2021} and \citet{chen2019} when ReLU or Softplus activation functions are used (see also Lemma \ref{lem:uat_conv_func}). Of course, other neural network architectures ensuring convexity in the input arguments could be employed as long as the universal approximation property can be proven, see Remark \ref{remark_kans_deep_sets}.

\subsection{Convex Signed Singular Value Neural Networks}
To approximate any frame-indifferent, isotropic polyconvex function under the incompressibility constraint, the same CSSV-NN framwork as in \citep{geuken2025}, but reduced by the determinant, is proposed and hence labeled as inc-CSSV-NN. This corresponds to the polyconvexity criterion given by Corollary~\ref{theo_singvalpolyconv_incompr}. The modified inc-CSSV-NN framework is summarized briefly below.\\
The basis is an ICNN that works with the elementary polynomials of the signed singular values as input. In order to account for the $\Pi_3$-invariance of the potential, $\Psi^{\textnormal{inc}}$ is given by a neural network potential reading
\begin{align}
\Psi^\nn_{\textnormal{ssv}} &= \dfrac{1}{24} \sum_{j = 1}^{24} \nn(\boldsymbol x^{(j)}), \label{nn_energy}\\
\textnormal{where} \quad \boldsymbol x^{(j)} & \in \textnormal{six permutations times four reflections of } \boldsymbol x^{(1)}, \textnormal{ cf.~Tab.~\ref{tab:input_permutations}}.
\end{align}
See Tab.~\ref{tab:input_permutations} for all combinations of $\boldsymbol x^{(j)}$. It should be emphasized that the same $\nn$ and therefore the same weights have to be used for every input permutation in order to guarantee $\Pi_3$-invariance. The overall inc-CSSV-NN architecture and all following neural network architectures are sketched in Fig.~\ref{fig_cssv_nn}. A universal approximation theorem for the underlying inc-CSSV-NN is proven in Section~\ref{sec_uat}. Such a universal approximation theorem cannot be derived for the other neural networks since they are only based on sufficient criteria.
The first Piola--Kirchhoff stress tensor defined by the network potential then reads
\begin{align}
\boldsymbol{P}^\nn_{\textnormal{ssv}} = \dfrac{\partial \Psi^\nn_{\textnormal{ssv}}}{\partial \boldsymbol F} = \dfrac{1}{24} \sum_{j = 1}^{24} \sum_{k = 1}^{3} \dfrac{\partial \nn(\boldsymbol x^{(j)})}{\partial \boldsymbol x^{(j)}} \cdot \dfrac{\partial \boldsymbol x^{(j)}}{\partial \nu_k} \, \dfrac{\partial \nu_k}{\partial \boldsymbol F}. \label{eq:stress}
\end{align}

Overall, the proposed $\Psi^\nn_{\textnormal{ssv}}$ represents a hyperelastic potential, is frame-indifferent, isotropic, polyconvex, accounts for incompressibility, guarantees a symmetric Cauchy stress $\boldsymbol \sigma$ and is thermodynamically consistent. Hence, the framework fulfills all desired constraints 1--5 and 7 given in Section~\ref{sec_fundamentals} exactly. Further constraints, such as the growth conditions \eqref{eq:growthcond_a} and \eqref{eq:growthcond_b} or energy normalization, can be easily enforced by adding respective terms to the potential
\begin{align}
\Psi = \Psi^\nn_{\textnormal{ssv}} + \Psi^{\textnormal{energy}} + \Psi^{\textnormal{growth}},
\end{align}
cf. \citep{geuken2025}. However, the growth conditions are not the focus of this study and furthermore, they are usually implicitly included in the training data. They are thus not considered further at this point and the interested reader is referred to \citep{linden2023, geuken2025} for details on how to design such terms.
Following \citep{linden2023, geuken2025}, the energy normalization is just a constant shift with respect to the undeformed configuration $\boldsymbol F = \boldsymbol I = \delta_{ij} \, \boldsymbol e_i \otimes \boldsymbol e_j$, i.e.,
\begin{align}
\Psi^{\textnormal{energy}} = -\Psi^\nn_{\textnormal{ssv}}(\boldsymbol I). \label{eq:energy_normalization}
\end{align}
The stress normalization is already included through the Lagrange multiplier in the extended stress, see Eq.~\eqref{eq_inc_stress_split}.

\begin{figure}[htbp!]
\centering
\includegraphics[width=\textwidth]{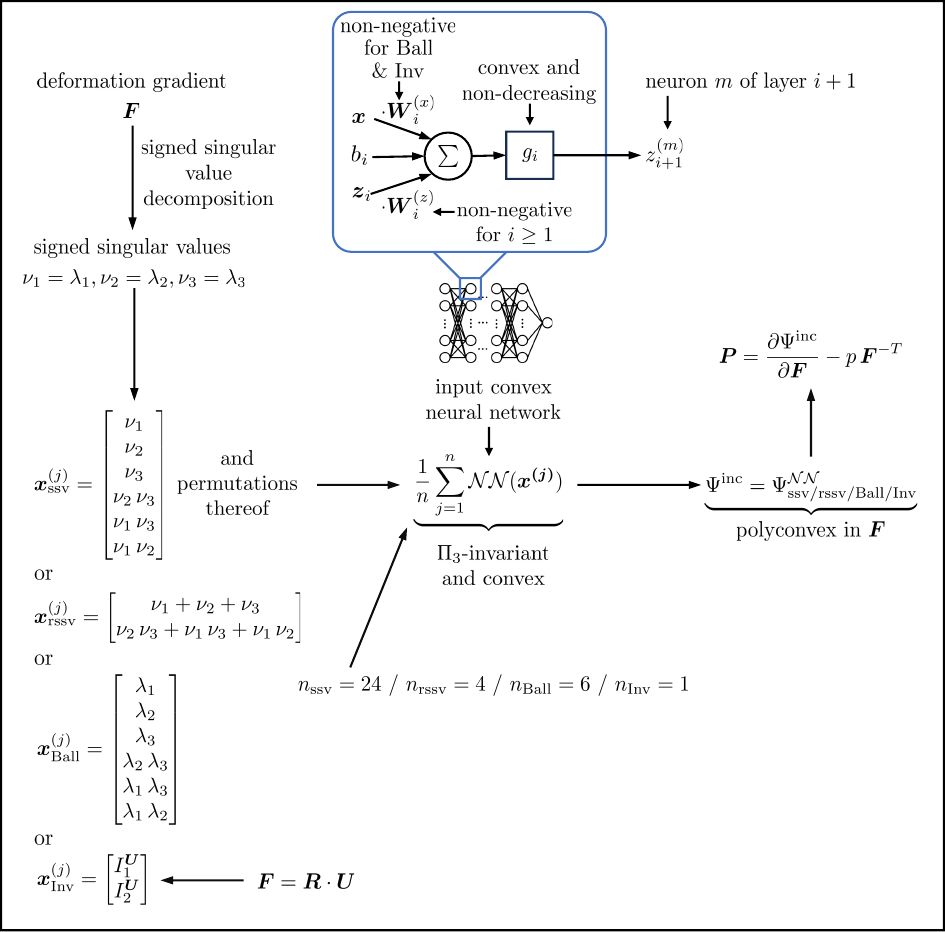}
\caption{Illustration of the incompressible convex signed singular value neural network (inc-CSSV-NN), the reduced inc-CSSV-NN, the inc-Ball NN and the inc-UInvar NN framework as a modification of the CSSV-NN illustration in \citep{geuken2025}. In a first step the signed singular values / principal stretches of the deformation gradient are computed. Depending on the network variant, the elementary polynomials of the signed singular values ($\nu_1, \nu_2, \nu_3$), the elementary symmetric polynomials ($\nu_1 + \nu_2 + \nu_3, \nu_2 \nu_3 + \nu_1 \nu_3 + \nu_1 \nu_2$), the elementary polynomials of the principal stretches ($\lambda_1, \lambda_2, \lambda_3$) or the invariants of $\bsym U$ ($I_1^{\bsym U}$, $I_2^{\bsym U}$) represent the input $\boldsymbol x^{(j)}$ of an input convex neural network $\nn$. The energy $\Psi^\nn_{\textnormal{ssv/rssv/Ball/Inv}} = \Psi^{\textnormal{inc}}$ is the average of $n$ network evaluations accounting for the required invariance according to Tab.~\ref{tab:input_permutations}, Eq.~\ref{nn_svals_red} or Tab.~\ref{tab:input_permutations_ball}. The stresses follow by derivation with respect to $\boldsymbol F$, cf. Eq.~\eqref{eq:stress}, and $p$ is determined from the surrounding pressure.} \label{fig_cssv_nn}
\end{figure}

\begin{remark} \label{remark_kans_deep_sets}
Kolmogorov--Arnold Networks (KANs) \citep{liu2025kan} currently emerge as promising alternatives to ICNNs. They were first applied to constitutive modeling by \citet{abdolazizi2025} and afterwards extended to account for polyconvexity in \citep{thako2025}. While the Kolmogorov--Arnold representation theorem provides a strong theoretical background, the approximation capabilities of KANs, especially under further constraints such as convexity, and limitations are yet to be explored more deeply and universal approximation theorems (for those specific cases) are to be developed. While this states an interesting, future challenge, it is beyond the scope of this work and not necessary for the present purpose since ICNNs already guarantee universal approximation.
Furthermore, the required Perm$(3)$-invariance could also be ensured via Deep Sets \citep{zaheer2018} to eventually save some computational time. However, it has not yet been investigated whether universal approximation can be achieved for the underlying constraints of $\Pi_3$-invariance and convexity. The proposed averaging technique in Eq.~\ref{nn_energy} is simple, effective, does not limit the network expressiveness and is computationally efficient since the network evaluations can be performed simultaneously.
\end{remark}

\subsection{Reduced Convex Signed Singular Value Neural Networks}
For neural networks based on the criterion in terms of the elementary symmetric polynomials (Eq.~\eqref{eq_redpolycrit_incompr}) the exact same architecture as for the inc-CSSV-NNs can be used except for the adapted input. The proposed energy then reads
\begin{align}
\Psi^\nn_{\textnormal{rssv}} &= \dfrac{1}{4} \sum_{j = 1}^{4} \nn(\boldsymbol x_{I,\nu}^{(j)}), \label{nn_svals_red}
\end{align}
with $\boldsymbol x_{I,\nu}^{(j)}$ according to Tab.~\ref{tab:input_permutations_rssv}. The stresses follow in the same manner as above and again further constraints can be easily enforced if desired. Due to the reduced input size and the smaller number of network evaluations, the model is computationally more efficient and less susceptible to local minima during optimization compared to the inc-/CSSV-NN. The expressiveness of this framework will be investigated by the numerical examples in section \ref{sec_examples}.

\begin{table}[h!]
\caption{Input permutations of the elementary symmetric polynomials of the signed singular values for the neural network employed in Eq.~\eqref{nn_svals_red}.}
\centering
\renewcommand{\arraystretch}{2}
\begin{tabular}{|l|l|l|l|}
\hline
$\boldsymbol x_{I,\nu}^{(1)} = (\nu_1 + \nu_2 + \nu_3, \nu_2 \nu_3 + \nu_1 \nu_3 + \nu_1 \nu_2)$ & $\boldsymbol x_{I,\nu}^{(2)} = (-\nu_1 - \nu_2 + \nu_3, -\nu_2 \nu_3 - \nu_1 \nu_3 + \nu_1 \nu_2)$ \\ \hline
$\boldsymbol x_{I,\nu}^{(3)} = (-\nu_1 + \nu_2 - \nu_3, -\nu_2 \nu_3 + \nu_1 \nu_3 -\nu_1 \nu_2)$ &
$\boldsymbol x_{I,\nu}^{(4)} = (\nu_1 - \nu_2 - \nu_3, \nu_2 \nu_3 - \nu_1 \nu_3 - \nu_1 \nu_2)$ \\ \hline
\end{tabular}%
\label{tab:input_permutations_rssv}
\end{table}

\subsection{Neural networks based on the improved Ball criterion formulated in principal stretches}
The general framework for neural network models based on Theorem \ref{theo_princstretchpolyconv_incompr} for polyconvexity (here referred to as Ball NN and inc-Ball NN in the incompressible case) is very similar to the one shown in the previous sections. The elementary polynomials of the principal stretches are the network input instead of the signed singular values. Then, only the respective invariance has to be changed and the non-decreasing property of Corollary \ref{theo_princstretchpolyconv_incompr} needs to be taken into account. The Perm$(3)$-invariance is guaranteed by setting
\begin{align}
\Psi^\nn_{\textnormal{Ball}} &= \dfrac{1}{6} \sum_{j = 1}^{6} \nn(\boldsymbol x_\lambda^{(j)}), \label{nn_ball}\\
\textnormal{where } \boldsymbol x_\lambda^{(j)} &= (\bsym B_j \cdot \bsym \lambda, \bsym B_j \cdot \bsym{\check{\lambda}}), \, \bsym B_j \in \textnormal{Perm}(3).
\end{align}
See Tab.~\ref{tab:input_permutations_ball} for all combinations of $\boldsymbol x_\lambda^{(j)}$. To account for the additional non-decreasing property -- in addition to $\boldsymbol W_{1:k-1}^{(\boldsymbol z)}$ -- also all $\boldsymbol W_{0:k-1}^{(\boldsymbol x)}$ have to be non-negative. This is quite restrictive for the expressiveness of ICNNs and comes on top of the fact that the improved Ball criterion is only sufficient for polyconvexity. However and for the same reasons as the reduced inc-/CSSV-NN, the model is computationally more efficient and less susceptible to local minima during optimization than the inc-/CSSV-NN but less efficient compared to the reduced inc-/CSSV-NN.

\begin{table}[h!]
\caption{Input permutations of the principal stretches for the neural network employed in Eq.~\eqref{nn_ball} based on the improved Ball criterion for polyconvexity.}
\centering
\renewcommand{\arraystretch}{2}
\resizebox{\textwidth}{!}{%
\begin{tabular}{|l|l|l|}
\hline
$\boldsymbol x_\lambda^{(1)} = \left( \lambda_1, \lambda_2, \lambda_3, \lambda_2 \lambda_3, \lambda_1 \lambda_3, \lambda_1 \lambda_2 \right)$ & $\boldsymbol x_\lambda^{(2)} = \left( \lambda_1, \lambda_3, \lambda_2, \lambda_2 \lambda_3, \lambda_1 \lambda_2, \lambda_1 \lambda_3 \right)$ & $\boldsymbol x_\lambda^{(3)} = \left( \lambda_2, \lambda_1, \lambda_3, \lambda_1 \lambda_3, \lambda_2 \lambda_3, \lambda_1 \lambda_2 \right)$ \\ \hline
$\boldsymbol x_\lambda^{(4)} = \left( \lambda_3, \lambda_1, \lambda_2, \lambda_1 \lambda_2, \lambda_2 \lambda_3, \lambda_1 \lambda_3 \right)$ & $\boldsymbol x_\lambda^{(5)} = \left( \lambda_2, \lambda_3, \lambda_1, \lambda_1 \lambda_3, \lambda_1 \lambda_2, \lambda_2 \lambda_3 \right)$ & $\boldsymbol x_\lambda^{(6)} = \left( \lambda_3, \lambda_2, \lambda_1, \lambda_1 \lambda_2, \lambda_1 \lambda_3, \lambda_2 \lambda_3 \right)$ \\ \hline
\end{tabular}%
}
\label{tab:input_permutations_ball}
\end{table}

\subsection{Polyconvex neural networks based on matrix invariants}
The formulation in terms of the matrix invariants already fulfills all neccessary invariances and thus the energy is simply given by
\begin{align}
\Psi^\nn_{\textnormal{Inv}} = \nn(I_1^{\boldsymbol U}, I_2^{\boldsymbol U}). \label{nn_invariants}
\end{align}
However, as for the neural network models based on the improved Ball criterion, all $\boldsymbol W_{0:k-1}^{(\boldsymbol x)}$ have to be non-negative to ensure the non-decreasing property. While this model, henceforth referred to as inc-UInvar NN, is the simplest and most efficient computationally, it inherently has the least expressive power.

\section{Universal approximation theorem for inc-CSSV-NNs for frame-indifferent, isotropic polyconvex functions under incompressibility}\label{sec_uat}

This section establishes a universal approximation result for the incompressible, frame-indifferent, isotropic polyconvex inc-CSSV-NN. Specifically, it is proven that any frame-indifferent, isotropic polyconvex function $\Psi$ under the incompressibility constraint can be approximated to arbitrary precision by an incompressible inc-CSSV-NN. The proof proceeds in three steps: First, the universal approximation of ICNNs for convex functions is recalled. Second, this result is extended to polyconvex functions expressed in terms of singular values under incompressibility and finally, by invoking the equivalence stated in Corollary \ref{theo_singvalpolyconv_incompr}, the proof for the incompressible, frame-indifferent, isotropic polyconvex case is completed.

\begin{lemma}[\textbf{Universal approximation theorem for convex functions \citep{geuken2025}}]\label{lem:uat_conv_func}
Let 
$\Omega 
 \subset \R^n$ be a convex and compact set and $f \colon \Omega \to \R$ be convex and, in particular, continuous. Then, for every $\varepsilon >0$ there exists an ICNN such that
\begin{align}
\sup\limits_{\boldsymbol{x} \in \Omega}\, 
\left| \nn(\boldsymbol{x}) - f (\boldsymbol{x})\right| < \varepsilon\,.
\end{align}
\end{lemma}

\begin{proof}
The proposed network is an ICNN by design and a proof for the universal approximation theorem is provided in \citep{huang2021} for ReLU and Softplus activation functions for $\Omega = [0, 1]^n$. Their result can be easily extended to arbitrary convex and compact sets.
\end{proof}

\begin{theorem}[\textbf{Universal approximation theorem for singular value polyconvex functions under incompressibility}]\label{prop:uat-psi}

Let $\Psi_{\textnormal{ssv}}^{\textnormal{inc}} \colon \R^6 \to \R_\infty$ be convex such that $\Psi_{\textnormal{ssv}}^{\textnormal{inc}}\circ \m^{\textnormal{inc}}$ is $\Pi_3$-invariant. Let $V\subset \R^3$ be a compact set where $\Psi_{\textnormal{ssv}}^{\textnormal{inc}}$ attains only finite values and $\boldsymbol{\nu} \in V$ fulfills $\nu_1 \nu_2 \nu_3 = 1$. Then, for every $\varepsilon > 0$, there exists a neural network $\Psi^\nn_{\textnormal{ssv}} \colon \R^6 \rightarrow \R_\infty$ of the form \eqref{nn_energy} such that
\begin{align}\label{eq:poly-Approximation-result}
    \sup\limits_{\boldsymbol{\nu} \in V} |\Psi^\nn_{\textnormal{ssv}}(m^{\textnormal{inc}}(\boldsymbol{\nu})) - \Psi_{\textnormal{ssv}}^{\textnormal{inc}}(m^{\textnormal{inc}}(\boldsymbol{\nu}))|
    &< \varepsilon\,.
\end{align}
\end{theorem}
\begin{proof}
Since $\Psi_{\textnormal{ssv}}^{\textnormal{inc}} \circ \m^{\textnormal{inc}}$ is finite on $V$ and $\Pi_3$-invariant, it is finite on the symmetric extension $\Pi_3(V) \coloneqq \{\bsym{\mathcal{P}}_j \cdot \bsym \nu \mid \bsym{\mathcal{P}}_j \in \Pi_3 \,,\, \bsym \nu \in V \}$ of $V$.
Moreover, since $\Psi_{\textnormal{ssv}}^{\textnormal{inc}}$ is convex and finite on 
$\m^{\textnormal{inc}} (\Pi_3(V)) = \{\m^{\textnormal{inc}}(\bsym \nu) \in \R^6 \mid \bsym \nu \in\Pi_3(V) \}$, it is also finite on the convex hull $\operatorname{conv}\m^{\textnormal{inc}}( \Pi_3(V))$ of $\m^{\textnormal{inc}}( \Pi_3(V))$.
Using Lemma \ref{lem:uat_conv_func}, $\Psi_{\textnormal{ssv}}^{\textnormal{inc}}$ is approximated on $\operatorname{conv}\m^{\textnormal{inc}}( \Pi_3(V))$ with $\nn$ by accuracy $\varepsilon$.
The $\Pi_3$-invariance of $\Psi_{\textnormal{ssv}}^{\textnormal{inc}} \circ \m^{\textnormal{inc}}$ implies $\Psi_{\textnormal{ssv}}^{\textnormal{inc}}(\m^{\textnormal{inc}}(\boldsymbol{\nu})) = \Psi_{\textnormal{ssv}}^{\textnormal{inc}}(\bsym{\mathcal{P}}_j \cdot \m^{\textnormal{inc}}(\boldsymbol{\nu}))$ for $\boldsymbol x^{(j)} = \bsym{\mathcal{P}}_j \cdot \boldsymbol x^{(1)}$ according to Tab.~\ref{tab:input_permutations}. Due to the group structure of $\Pi_3$, the equality holds also for $\bsym{\mathcal{P}}_j^{-1}$. Now, utilizing the $\triangle$-inequality for subadditivity, $\Pi_3$-invariance and the approximation by Lemma \ref{lem:uat_conv_func}, one can conclude:
\begin{align*}
\sup\limits_{\boldsymbol{\nu} \in V} |\Psi^\nn_{\textnormal{ssv}}(m^{\textnormal{inc}}(\boldsymbol{\nu})) - \Psi_{\textnormal{ssv}}^{\textnormal{inc}}(m^{\textnormal{inc}}(\boldsymbol{\nu}))|
=
&\sup\limits_{\boldsymbol{\nu} \in V} \dfrac{1}{24} \left| \sum_{j = 1}^{24} \left(\nn (\underbrace{\boldsymbol x^{(j)}}_{\bsym{\mathcal{P}}_j \cdot \boldsymbol x^{(1)})}(\boldsymbol\nu) - \Psi_{\textnormal{ssv}}^{\textnormal{inc}}(\m^{\textnormal{inc}}(\boldsymbol\nu))\right) \right| \\
\stackrel{\triangle}{\le}
&\dfrac{1}{24} \sum_{j = 1}^{24} \sup\limits_{\boldsymbol{\nu} \in V} \left| \nn (\boldsymbol x^{(j)}(\boldsymbol\nu)) - \Psi_{\textnormal{ssv}}^{\textnormal{inc}}(\m^{\textnormal{inc}}(\boldsymbol\nu)) \right| \\
\stackrel{}{=}
&\dfrac{1}{24} \sum_{j = 1}^{24} \sup\limits_{\m^{\textnormal{inc}}(\boldsymbol{\nu}) \in \bsym{\mathcal{P}}_j \cdot \m^{\textnormal{inc}}(V)} \left| \nn (m^{\textnormal{inc}}(\boldsymbol\nu)) - \Psi_{\textnormal{ssv}}^{\textnormal{inc}}(\bsym{\mathcal{P}}_j^{-1} \cdot \m^{\textnormal{inc}}(\boldsymbol\nu)) \right| \\
\stackrel{\Pi_3}{=}
&\dfrac{1}{24} \sum_{j = 1}^{24} \sup\limits_{\boldsymbol{\nu} \in V} \left| \nn (m^{\textnormal{inc}}(\boldsymbol\nu)) - \Psi_{\textnormal{ssv}}^{\textnormal{inc}}(\m^{\textnormal{inc}}(\boldsymbol\nu)) \right| \\
\stackrel{\textnormal{Lemma \ref{lem:uat_conv_func}}}{<}\hspace{-0.43cm}
&\hspace{0.43cm} \dfrac{1}{24} \sum_{j = 1}^{24} \varepsilon = \varepsilon\,.
\end{align*}
\end{proof}
The universal approximation for frame-indifferent, isotropic polyconvex functions $\Psi^{\textnormal{inc}}$ under incompressibility then follows from combining Theorem \ref{prop:uat-psi} with the equivalence in Corollary \ref{theo_singvalpolyconv_incompr}.
\begin{remark}
Universal approximation can also be guaranteed in the case of energy normalization~\eqref{eq:energy_normalization} since $\Psi^{\textnormal{energy}}$ is just a constant, see also \citep{geuken2025}.
\end{remark}

\section{Algorithmic implementation} \label{sec_implementation}
All neural network models were implemented in \textit{Python} using \textit{TensorFlow} and \textit{SciPy} based on the implementation of CSSV-NNs \citep{geuken2025}. 

\subsection{Network architecture} \label{sec_architecture}
For all network variants, the Softplus function $g^{\mathrm{SP}}(x) = \textnormal{ln}(1+\textnormal{exp}(x))$ was selected as the activation function for all layers except the output layer, which employs linear activation. Both functions are convex, non-decreasing and smooth and therefore suitable for ICNNs and the polyconvexity criteria. Beyond that, this choice guarantees universal approximation. For inc-/CSSV-NNs the weights $\boldsymbol W_{k-1}^{(\boldsymbol x)}$ were set to zero, since the input arguments cancel out due the required symmetry in the case of a last linear layer. This modification does not diminish the expressiveness of CSSV-NNs or its universal approximation ability.

Within the numerical examples, eight network configurations varying in width and depth were conducted. Their size is denoted by "size of input vector -- number of neurons of hidden layer 1 -- ... number of neurons of hidden layer $(k-1)$ -- size of final layer $k$": 
\begin{multicols}{4}
\begin{enumerate}
\item $a$--4--1
\item $a$--8--1
\item $a$--12--1
\item $a$--4--2--1
\item $a$--8--4--1
\item $a$--12--8--1
\item $a$--8--4--4--1
\item $a$--12--8--4--1
\end{enumerate} 
\end{multicols} 
\noindent wherein $a$ is the input vector size of each neural network architecture given in Tab.~\ref{tab:input_size}. The selection of these architectures was guided by preceding hyperparameter studies, the work of \citet{geuken2025} and the demonstrated effectiveness of deeper ICNNs in \citep{chen2019}.

\begin{table}[h!]
\caption{Input vector size $a$ of all the different neural network architectures. The input size is reduced by one in the incompressible case compared to the compressible one due to the removed determinant.}\label{tab:input_size}
\centering
\begin{tabular}{|l|c|l|c|}
\hline
& \phantom{x}a\phantom{x} &  & \phantom{x}a\phantom{x} \\ \hline
CSSV-NN & 7 & inc-CSSV-NN & 6 \\ \hline
Reduced CSSV-NN & 3 & Reduced inc-CSSV-NN & 2 \\ \hline
Ball NN & 7 & inc-Ball NN & 6 \\ \hline
UInvar NN & 3 & inc-UInvar NN & 2 \\ \hline 
\end{tabular}%
\end{table}

\subsection{Sobolev training}\label{sec_sobolev_training}
The networks are trained in Sobolev space by strain--stress tuples
\begin{align}
\mathcal{D} = \left\{ \left(\boldsymbol{F}^{(1)},\boldsymbol{P}^{(1)}\right), \, \left(\boldsymbol F^{(2)}, \boldsymbol{P}^{(2)}\right), \, ... \, , \, \left(\boldsymbol F^{(n)}, \boldsymbol{P}^{(n)}\right) \right\}
\end{align}
based on the mean squared error 
\begin{align}
\textnormal{MSE} = \dfrac{1}{n} \sum_{i = 1}^{n} \left\Vert \boldsymbol{P}^{(i)}-\boldsymbol{P}^\nn \left(\boldsymbol{F}^{(i)}; \, \theta\right) \right\Vert^2 \label{eq_msestress}
\end{align}
with the Froebenius norm $\Vert \bullet \Vert$ as it is nowadays common practice in constitutive modeling. This formulation measures deviations in stress space, thus offering a natural physical interpretation. While alternative loss definitions are possible (e.g., based on strain--energy data or a mixture of strain--energy and strain--stress), training exclusively on stress data produced the most reliable outcomes and aligns well with the availability of experimental measurements. Note that the inc-CSSV-NN framework can also be applied to approximate polyconvex hulls in the same manner as described in \citep{geuken2025}.

\section{Numerical examples} \label{sec_examples}
In the following, several numerical examples are conducted to analyze the expressiveness, accuracy, simplicity and efficiency of the proposed models. In a first step, the models are trained on synthetic data from four well-established, classical models. Then, they are trained on Treloar's experimental data \citep{treloar1944} to highlight their practical applicability. Finally, an explicitly constructed energy function is provided that cannot be approximated by neural networks based on the improved Ball criterion for polyconvexity.

\subsection{Classical models}
All proposed, incompressible neural network models were trained on synthetic data obtained from the following classical, incompressible, frame-indifferent, isotropic polyconvex material models:
\begin{enumerate}
\item Neo--Hooke: \begin{align}
\Psi^{\textnormal{inc}}_{\textnormal{NH}} = c_1 \, (I_1^{\bsym C} - 3) \label{eq:neo_hooke}
\end{align}
\item Mooney--Rivlin: \begin{align}
\Psi^{\textnormal{inc}}_{\textnormal{MR}} = c_1 \, (I_1^{\bsym C} - 3) + c_2 \, (I_2^{\bsym C} - 3) \label{eq:mooney_rivlin}
\end{align}
\item Gent: \begin{align}
\Psi^{\textnormal{inc}}_{\textnormal{Gent}} = - c_1 \, (I_m - 3) \, \textnormal{ln}\bigg(1-\dfrac{I_1^{\bsym C} - 3}{I_m - 3}\bigg) \label{eq:gent}
\end{align}
\item Arruda--Boyce: \begin{align}
\Psi^{\textnormal{inc}}_{\textnormal{AB}} = \, &c_1 \, \bigg(\dfrac{1}{2} \, (I_1^{\bsym C} - 3) + \dfrac{1}{20 \, N} \, ((I_1^{\bsym C})^2 - 9) + \dfrac{11}{1050 \, N^2} \, ((I_1^{\bsym C})^3 - 27) \label{eq:arruda}\\
&+ \dfrac{19}{7000 \, N^3} \, ((I_1^{\bsym C})^4 - 81) + \dfrac{519}{673750 N^4} \, ((I_1^{\bsym C})^5 - 243)\bigg) \nonumber
\end{align}
\end{enumerate}
The material parameters were chosen such that all models linearize to the same Hooke model for small strains and are given in Tab.~\ref{tab:mat_params}. The Arruda--Boyce energy here is a fifth-order approximation in $I_1^{\bsym C}$ of the original energy as described in \citep{arrudaboyce1993}. The following load cases have been used for training (load direction $F_{ij}$, format $[$start value : increment : end value$]$):
\begin{itemize}
\item no load $F_{ij} = \delta_{ij}$ (= $1$ if $i=j$, 0 else) 
\item uniaxial compression/tension $F_{11} \in [0.5:0.1:2.5]$
\item biaxial compression/tension $F_{11} = F_{22} \in [0.7:0.1:2.0]$
\item pure shear (planar compression/tension) $F_{11} \in [0.5:0.1:2.0], F_{22} = 1$
\end{itemize}
leading to 52 data points (with uniform distribution within each load case). $F_{33}$ follows from incompressibility and the Lagrange multiplier $p$ from the constraint $P_{33} = 0$. A separate evaluation with validation and test data sets is omitted for the noise-free synthetic data base that is directly extracted from the ground truth reference. This does not limit any of the following statements. For each combination of material model, neural network model (inc-CSSV-NN, reduced inc-CSSV-NN, inc-Ball NN and inc-UInvar NN) and architecture (1--8), 30 instances with different random weight initializations were trained using \textit{SciPy’s minimize} function, which performs a constrained L-BFGS-B optimization. The number of maximum iterations was set to 1000. However, nearly every optimization stopped earlier based on \textit{SciPy’s} default criteria.

\begin{table}[h!]
\caption{Employed material parameters for the classical models corresponding to a shear modulus of $\mu = 2\,\textnormal{MPa}$ and a Young's Modulus of $E = 6\,\textnormal{MPa}$ for the linearized model.}
\centering
\begin{tabular}{|l|c|c|}
\hline
Neo--Hooke & $c_1 = 1.0 \,\textnormal{MPa}$ & - \\ \hline Mooney--Rivlin & $c_1 = 0.8 \,\textnormal{MPa}$ & $c_2 = 0.2 \,\textnormal{MPa}$ \\ \hline
Gent                                & $c_1 = 1.0 \,\textnormal{MPa}$ & $I_m = 30$                      \\ \hline
Arruda--Boyce                        & $c_1 = 1.7 \,\textnormal{MPa}$ & $N = 4$ \\ \hline
\end{tabular}%
\label{tab:mat_params}
\end{table}

The inc-CSSV-NN, the reduced inc-CSSV-NN and the inc-Ball NN capture all four energies perfectly as can be seen in Figs.~\ref{fig_neohooke}--\ref{fig_arruda} and further supported by the low mean squared errors of the best-performing instances given in Tab.~\ref{tab:mses_classical}. Interestingly, the models also extrapolate most of the trained load paths \eqref{eq:neo_hooke}--\eqref{eq:arruda} perfectly. This indicates that they indeed learned the underlying analytic relations for the moderately extended strain domains in the present settings. In contrast, the inc-UInvar NN shows much higher mean squared errors (Tab.~\ref{tab:mses_classical}) and fails to learn the uniaxial tension behavior of all material models except the Mooney--Rivlin model accurately. In addition, the extrapolation behavior is worse compared to the other three neural network models. From that it can be concluded that a neural network based on invariants of $\bsym U$ does not approximate $I_1^{\bsym C}$ very well and thus has limited expressiveness and accuracy. Clearly, a neural network based on invariants of $\bsym C$ would perform better in that case, however and as stated before, it was already demonstrated in \citep{geuken2025} that other energy functions are unattainable by such neural networks. Thus, it is not clear which criterion is stronger. In the 2D case, the criterion connected to $\bsym U$ includes the one of $\bsym C$ since the invariants of $\bsym C$ can be expressed as convex and partly non-decreasing functions of the invariants of $\bsym U$ and therefore is more general. However, this does not apply for the 3D case. These results show that formulations via invariants -- though being easy to implement and efficient computational wise -- are not optimal due to their restricted expressiveness compared to the other approaches. 
\begin{figure}[htbp!]
\centering
\caption*{\textbf{Neo--Hooke}}
\includegraphics[width=\textwidth]{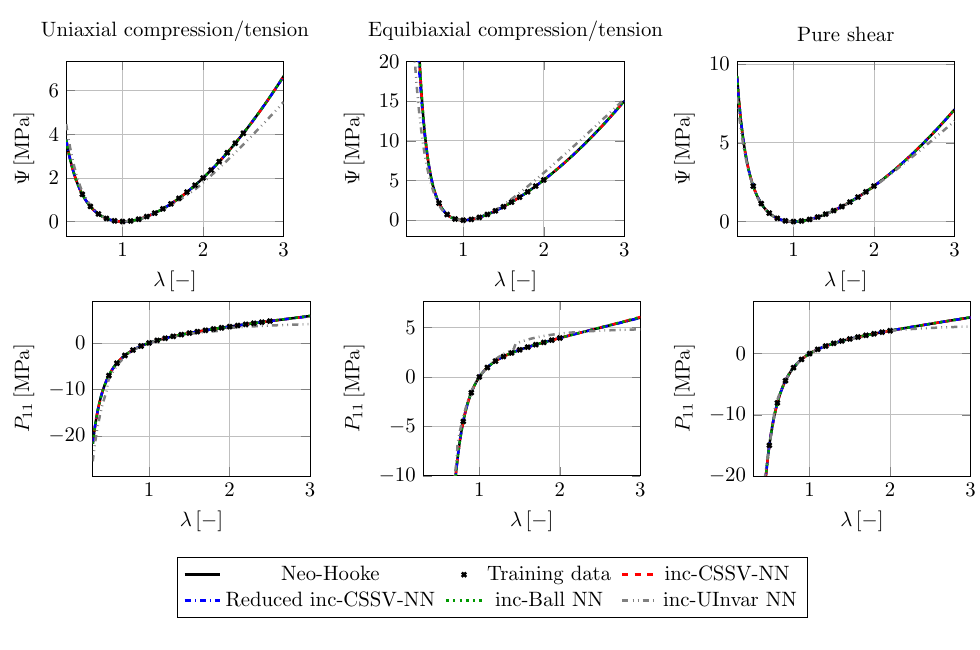}
\caption{Approximation of the Neo--Hooke model \eqref{eq:neo_hooke}: Energy and stress predictions of the inc-CSSV-NN, reduced inc-CSSV-NN, inc-Ball NN and inc-UInvar NN for uniaxial ($F_{11} = \lambda$) and equibiaxial ($F_{11} = F_{22} = \lambda$) compression and tension as well as pure shear/planar compression/tension ($F_{11} = \lambda$, $F_{22}= 1$).}
\label{fig_neohooke}
\end{figure}

\begin{figure}[htbp!]
\centering
\caption*{\textbf{Mooney--Rivlin}}
\includegraphics[width=\textwidth]{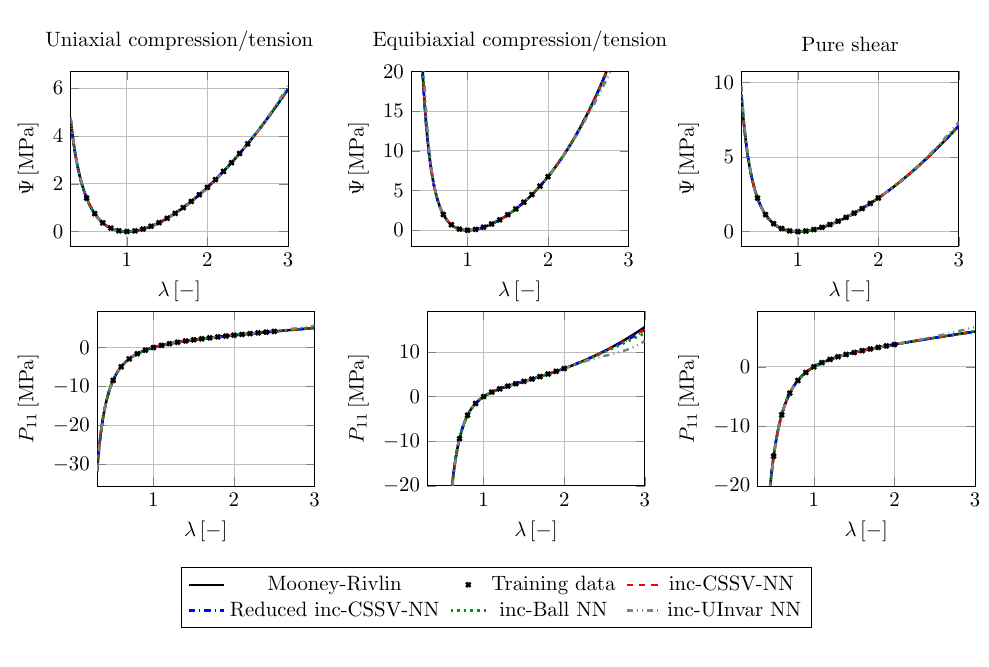}
\caption{Approximation of the Mooney--Rivlin model \eqref{eq:mooney_rivlin}: Energy and stress predictions of the inc-CSSV-NN, reduced inc-CSSV-NN, inc-Ball NN and inc-UInvar NN for uniaxial ($F_{11} = \lambda$) and equibiaxial ($F_{11} = F_{22} = \lambda$) compression and tension as well as pure shear/planar compression/tension ($F_{11} = \lambda$, $F_{22}= 1$).}
\label{fig_mooneyrivlin}
\end{figure}

\begin{figure}[htbp!]
\centering
\caption*{\textbf{Gent}}
\includegraphics[width=\textwidth]{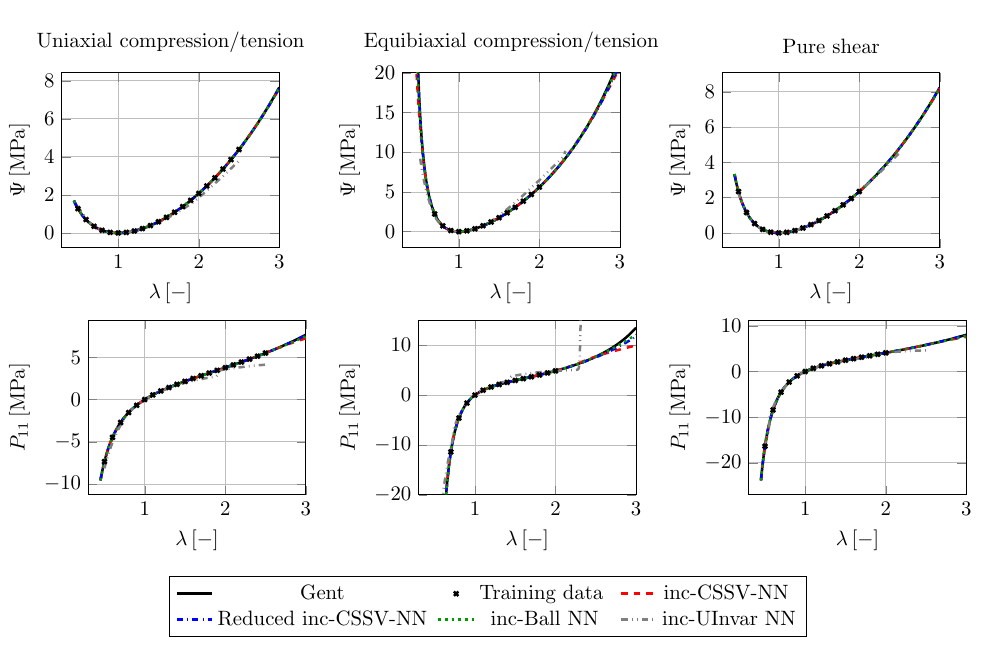}
\caption{Approximation of the Gent model \eqref{eq:gent}: Energy and stress predictions of the inc-CSSV-NN, reduced inc-CSSV-NN, inc-Ball NN and inc-UInvar NN for uniaxial ($F_{11} = \lambda$) and equibiaxial ($F_{11} = F_{22} = \lambda$) compression and tension as well as pure shear/planar compression/tension ($F_{11} = \lambda$, $F_{22}= 1$).}
\label{fig_gent}
\end{figure}

\begin{figure}[htbp!]
\centering
\caption*{\textbf{Arruda--Boyce}}
\includegraphics[width=\textwidth]{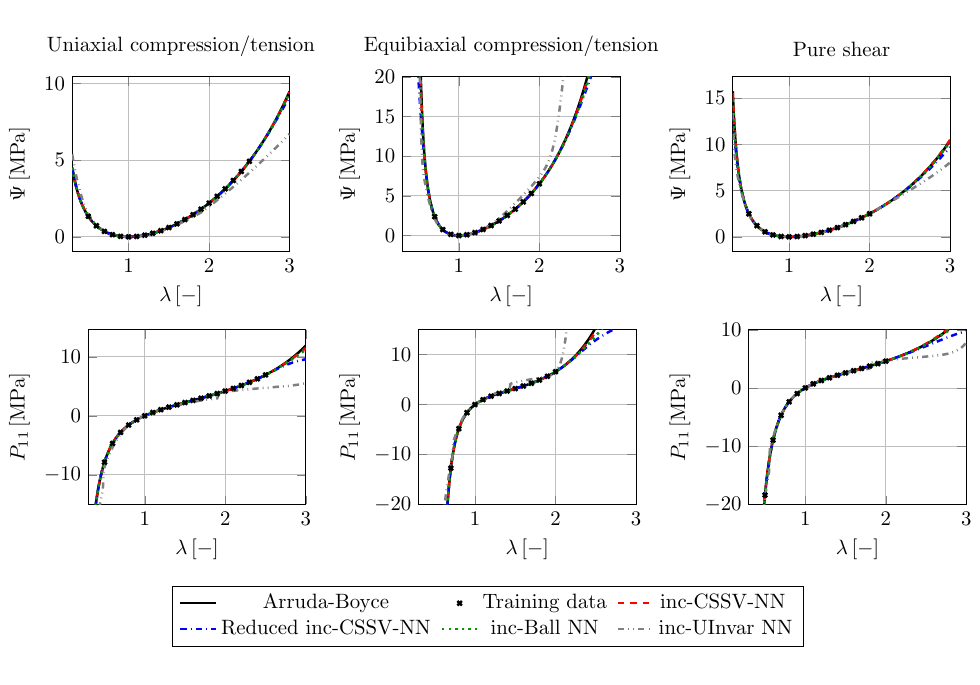}
\caption{Approximation of the Arruda--Boyce model \eqref{eq:arruda}: Energy and stress predictions of the inc-CSSV-NN, reduced inc-CSSV-NN, inc-Ball NN and inc-UInvar NN for uniaxial ($F_{11} = \lambda$) and equibiaxial ($F_{11} = F_{22} = \lambda$) compression and tension as well as pure shear/planar compression/tension ($F_{11} = \lambda$, $F_{22}= 1$).}
\label{fig_arruda}
\end{figure}

\begin{table}[h!]
\caption{Achieved mean squared errors of the inc-CSSV-NN, reduced inc-CSSV-NN, inc-Ball NN and inc-UInvar NN for the training data sets of the classical models as well as the inc-Mielke model.}\label{tab:mses_classical}
\centering
\begin{tabular}{|l|c|c|c|c|}
\hline
              & inc-CSSV-NN & Reduced inc-CSSV-NN & inc-Ball NN & inc-UInvar NN\\ \hline
Neo--Hooke     & 5.14e-09 & 4.94e-09 & 1.42e-07 & 3.57e-02 \\ \hline
Mooney--Rivlin & 5.66e-09 & 9.15e-09 & 9.74e-08 & 2.95e-04 \\ \hline
Gent          & 8.39e-09 & 1.25e-08 & 1.46e-07 & 4.16e-02 \\ \hline
Arruda--Boyce  & 1.06e-07 & 2.88e-08 & 4.11e-07 & 5.98e-02 \\ \hline
inc-Mielke    & 9.73e-05 & 6.63e-05 & 6.95e-03 & 7.44e-02 \\ \hline
\end{tabular}%
\end{table}

\subsection{Treloar data}
From a practical perspective, training with experimental data is also very relevant. Therefore, all neural networks were trained on Treloar's well-known experimental data of (nearly incompressible) vulcanized rubber \citep{treloar1944} to investigate their practical applicability. The data set includes uniaxial tension, equibiaxial tension and pure shear loading. As in \citep{geuken2025}, the values of the 56 strain--stress tuples were taken from \citep{steinmann2012} and the tuples were split into a training, validation and test data set. The training data set comprises 41, the validation set 10 and the test set 5 randomly selected data points. However, since the pure shear state is only an assumption based on a uniaxial tensile test of a wide specimen, the stress component in width direction is unknown and therefore removed from the data set. Once more, for each combination of neural network model and architecture, 30 instances with different random weight initializations were trained using \textit{SciPy’s minimize} function. In the end, the model with the lowest mean squared error in the validation data was chosen.

Again, the inc-CSSV-NN, the reduced inc-CSSV-NN and the inc-Ball NN fit the material behavior perfectly, while the inc-UInvar NN fails to do so, cf. Fig.~\ref{fig_treloar}. This is also evident in the mean squared errors shown in Tab.~\ref{tab:mses_treloar}. It should be mentioned, that the first three neural network models perform magnitudes better than all 14 classical models investigated in \citep{steinmann2012} highlighting the superior expressiveness and accuracy of neural network models. Unfortunately, an exact
error comparison is not possible because the stress components included in the error \citep{steinmann2012} are not explicitly accessible. Interestingly, the test set mean squared errors are lower than the training and validation errors. This can be justified by the data distribution: the randomly selected, few test data points lie all in low stretch and stress regimes, which inherently show lower squared errors. In these regimes even the inc-UInvar NN provides a really good approximation.

So far, no significant difference between the inc-CSSV-NN, the reduced inc-CSSV-NN and the inc-Ball NN regarding the accuracy and expressiveness can be identified. For this reason, a final example is examined.
\begin{figure}[htbp!]
\centering
\caption*{\textbf{Treloar's experimental data}}
\includegraphics[width=\textwidth]{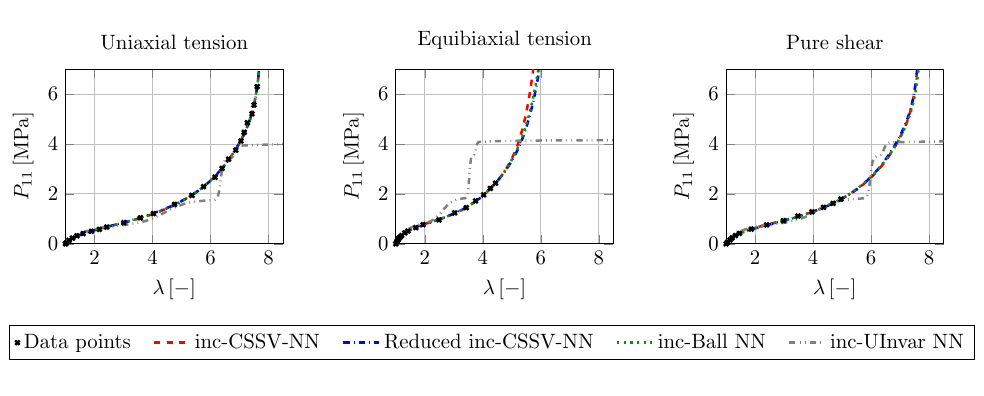}
\caption{Approximation of Treloar's experimental data: Stress predictions of the inc-CSSV-NN, reduced inc-CSSV-NN, inc-Ball NN and inc-UInvar NN for uniaxial ($F_{11} = \lambda$) and equibiaxial ($F_{11} = F_{22} = \lambda$) tension as well as pure shear/planar tension ($F_{11} = \lambda$, $F_{22}= 1$).}
\label{fig_treloar}
\end{figure}

\begin{table}[h!]
\caption{Achieved mean squared errors of the inc-CSSV-NN, reduced inc-CSSV-NN, inc-Ball NN and inc-UInvar NN for the training, validation and test data set of Treloar's experimental data.}
\label{tab:mses_treloar}
\centering
\begin{tabular}{|l|c|c|c|c|}
\hline
                    & inc-CSSV-NN & Reduced inc-CSSV-NN & inc-Ball NN & inc-UInvar NN \\ \hline
Training data set   & 3.65e-05 & 3.98e-05 & 3.89e-05 & 6.51e-02 \\ \hline
Validation data set & 3.35e-04 & 3.56e-04 & 3.70e-04 & 7.07e-02 \\ \hline
Test data set       & 1.07e-05 & 1.30e-05 & 1.35e-05 & 2.84e-04 \\ \hline
\end{tabular}%
\end{table}

\subsection{Mielke energy} \label{sec_mielke_energy}
Together with necessary and sufficient conditions for polyconvexity of frame-indifferent isotropic energy functions, \citet{mielke2005} proposed the energy
\begin{align}
\Psi_{\textnormal{Mielke}} = \textnormal{max} \left( \vert \nu_1 - \nu_2 \nu_3 \vert, \vert \nu_2 - \nu_1 \nu_3 \vert, \vert \nu_3 - \nu_1 \nu_2 \vert \right), \label{eq:mielke_energy}
\end{align}
and showed that it is polyconvex, despite being non-convex in $\bsym \nu$ and, in particular, not monotone in any $\nu_i$. Especially due to the latter, it seems reasonable to assume that this energy cannot be approximated by neural networks that are based on the improved Ball criterion for polyconvexity, which is investigated in this section. Since the energy \eqref{eq:mielke_energy} is not differentiable, a smooth, polyconvex approximation given by
\begin{align}
\Psi^{\textnormal{approx}}_{\textnormal{Mielke}} = \dfrac{1}{a} \, \textnormal{log} \left( \textnormal{cosh} \left( a \, (\nu_1 - \nu_2 \nu_3) \right) + \textnormal{cosh} \left( a \, (\nu_2 - \nu_1 \nu_3) \right) + \textnormal{cosh} \left( a \, (\nu_3 - \nu_1 \nu_2) \right) \right), \label{eq:mielke_approx2} 
\end{align}
wherein $a$ is a numerical parameter controlling the smoothness, is applied instead. It will be first used as an incompressible model (only deformation gradients $\bsym F$ satisfying $J = \det \bsym{F} = 1$ and adapted stresses) and then as originally proposed by \citet{mielke2005} as a compressible model.

\subsubsection{Incompressible Mielke energy} \label{sec_mielke_incompr}

In the following
\begin{align}
\Psi^{\textnormal{approx, inc}}_{\textnormal{Mielke}} = \dfrac{1}{a} \, \textnormal{log} \left( \textnormal{cosh} \left( a \, (\nu_1 - \nu_2 \nu_3) \right) + \textnormal{cosh} \left( a \, (\nu_2 - \nu_1 \nu_3) \right) + \textnormal{cosh} \left( a \, (\nu_3 - \nu_1 \nu_2) \right) \right) \label{eq:mielke_approx}\\ \forall \, \bsym \nu \textnormal{ with } \nu_1 \, \nu_2 \, \nu_3 = 1, \nonumber
\end{align}
is referred to as the inc-Mielke model. The incompressibility constraint allows a parametrization of the energy and stresses in two independent principal stretches and thus the neural network models were trained on states $[F_{11},F_{22}] \in [0.5:0.075:2]^2$, leading to 441 data points, with $a=10$ and $F_{33}$ following from incompressibility. The training procedure was the same as for the classical models.

This time, only the inc-CSSV-NN and the reduced inc-CSSV-NN are able to reproduce the material model perfectly, cf. Fig.~\ref{fig_mielke} and Tab.~\ref{tab:mses_classical}. Their lowest mean squared errors are 9.73e-05 for the inc-CSSV-NN and 6.63e-05 for the reduced version. In contrast, the lowest mean squared error achieved by the inc-Ball NNs (6.95e-03) is 71 times higher than the lowest inc-CSSV-NN error and 104 times higher than the lowest reduced inc-CSSV-NN error. The inc-Ball NN is, e.g., not able to predict the stress component $P_{11}$ for compressive strains in $11$-direction while stretched in $22$-direction, see Fig.~\ref{fig_mielke} right, or decreasing stresses in the pure shear case ($\lambda_2 = 1$, Fig.~\ref{fig_mielke} middle and right) accurately. This behavior is expected from the respective polyconvexity criterion. The inc-UInvar NN performs even worse (lowest mean squared error 7.44e-02) and shows only a smooth, but far off approximation of the stresses (Fig.~\ref{fig_mielke}). The details of the material behavior prescribed by Eq.~\eqref{eq:mielke_approx} are not captured.

Finally, it can be concluded from this example that the expressivity of the inc-Ball NN is indeed limited. This indicates that the improved version of Ball's polyconvexity criterion and thus also Ball's original criterion is not necessary in the incompressible case. In contrast, the reduced inc-CSSV-NN approximated all incompressible energies perfectly. It hence does not contradict the ?-marked relation in Fig.~\ref{fig_crit_overview} regarding the necessity of the reduced criterion based on the elementary symmetric polynomials of the signed singular values in the incompressible case. However, further mathematical investigations are required and rigorous proofs of these statements remain an open challenge.

\begin{figure}[htbp!]
\centering
\caption*{\textbf{inc-Mielke}}
\vspace{-1em}
\includegraphics[width=\textwidth]{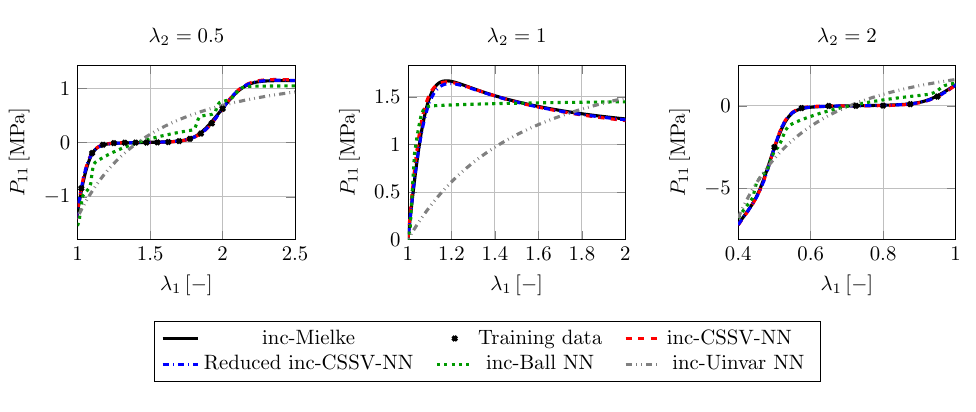}
\caption{Approximation of the inc-Mielke model \eqref{eq:mielke_approx}: Representative stress over stretch ($\lambda_1 = F_{11}$) predictions of the inc-CSSV-NN, reduced inc-CSSV-NN, inc-Ball NN and inc-UInvar NN for different, fixed stretches in the perpendicular direction $\lambda_2 = F_{22}$. The inc-Ball NN is, e.g., not able to predict the decreasing stress for planar tension ($\lambda_2 = 1$) and only provides a constant approximation (middle).}
\label{fig_mielke}
\end{figure}

\subsubsection{Compressible Mielke energy} \label{sec_mielke_compr}
This section particularly aims to compare the performance of the compressible CSSV-NN, the reduced CSSV-NN and the Ball NN using the smooth approximation of the compressible energy proposed by \citet{mielke2005}. The UInvar NN was not investigated here, as this framework already showed poor performance for the previous incompressible energies. Notably, the original energy \eqref{eq:mielke_energy} can be expressed in terms of ReLu functions and thus by a six layer CSSV-NN with ReLu activation. To achieve this, six linear input combinations are defined first
\begin{align*}
h_1 = \nu_1 - \nu_2 \, \nu_3 \quad\quad\quad
h_2 = \nu_2 \, \nu_3 - \nu_1 \quad\quad\quad
h_3 = \nu_2 - \nu_1 \, \nu_3& \\
h_4 = \nu_1 \, \nu_3 - \nu_2 \quad\quad\quad
h_5 = \nu_3 - \nu_1 \, \nu_2 \quad\quad\quad
h_6 = \nu_1 \, \nu_2 - \nu_3&.
\end{align*}
Then the CSSV-NN exactly representing \eqref{eq:mielke_energy} is given by 
\begin{equation*}
L_0 = \boldsymbol x^{(j)} 
\end{equation*}
\begin{equation*}
L_1 = \relu(h_1 - h_2) + h_2
\end{equation*}
\begin{equation*}
L_2 = \relu(L_1 - h_3) + h_3
\end{equation*}
\begin{equation*}
L_3 = \relu(L_2 - h_4) + h_4
\end{equation*}
\begin{equation*}
L_4 = \relu(L_3 - h_5) + h_5
\end{equation*}
\begin{equation*}
L_5 = \relu(L_4 - h_6) + h_6
\end{equation*}
\begin{equation}
\Psi^\nn_{\textnormal{ssv}} = \dfrac{1}{24} \sum_{j = 1}^{24} L_5.
\end{equation}
Motivated by this, various deep (up to ten layers) and wide (up to 500 neurons per layer) ReLu networks with random weight initialization were trained on strain--energy tuples generated from the original energy \eqref{eq:mielke_energy} as an additional study. Surprisingly, the optimizer consistently failed to recover the analytical solution. Accordingly, even if the space spanned by the neural network includes the analytical solution, the applied optimization algorithm might not find it. This observation is also relevant for interpreting the results of the neural networks trained on the smooth approximation \eqref{eq:mielke_approx2} presented below.

As in the previous section, $a$ was set to 10 for the smooth approximation. The dataset was generated by sampling diagonal deformation gradients of the form
\begin{equation}
\bsym F= \textbf{diag}(\lambda_1,\lambda_2,\lambda_3),
\end{equation}
where each principal stretch $\lambda_i$ was selected from the uniformly discretized range $[0.6:0.2:2.0]$. To avoid redundant permutations, only unique stretch combinations satisfying $\lambda_1 \geq \lambda_2 \geq \lambda_3$ were considered, resulting in 120 distinct deformation states. The training procedure was the same as for the classical models except that the number of random weight initializations per architecture was increased to 60 and the number of maximum iterations was set to 10000 since more instances showed a poor performance and slower convergence compared to all other energies presented in \citep{geuken2025} and this work. For this reason, all models were fine-tuned by \textit{SciPy’s} SLSQP method after training with the L-BFGS-B. In addition to the eight networks sizes presented in Section~\ref{sec_architecture} and motivated by the analytical construction above, various deeper (up to ten layers) and wider (up to 500 neurons per layer) networks were trained with the same settings but did not lead to better results.

The best CSSV-NN achieved a mean squared error of 1.67e-04, while the best reduced CSSV-NN led to a mean squared error of 2.80e-03 and the best Ball NN to 1.39e-02. Most of the stress predictions made by the Ball NN are more than 10\% off from the respective training data reference and so are many reduced CSSV-NN predictions. Even the CSSV-NN shows some predictions that are more than 10\% off, see Fig.~\ref{fig_mielke_prediction} (gray areas). Nevertheless, the stress over stretch behavior -- here $P_{11}$ (Fig.~\ref{fig_mielke_stress11}) and $P_{22}$ (Fig.~\ref{fig_mielke_stress22}) over $F_{11}$ for various, but fixed $F_{22}$ and $F_{33}$ combinations -- is captured really well by the CSSV-NN. The larger errors occur at the boundary of the training data regime (Fig.~\ref{fig_mielke_stress22} bottom middle for $F_{11}$=0.6) or as a smooth approximation of the rapid stress change of the reference model (Fig.~\ref{fig_mielke_stress22} top middle). The $P_{11}$ stresses are captured perfectly by the CSSV-NN (Fig.~\ref{fig_mielke_stress22}). The reduced CSSV-NN and the Ball NN show some larger deviations, especially for the $P_{22}$ stresses (Fig.~\ref{fig_mielke_stress22}). Overall, the remaining error of the CSSV-NN is most likely attributable to the optimizer converging to local minima, since the smooth target energy \eqref{eq:mielke_approx2} can, in principle, be approximated arbitrarily well by the chosen network architecture and the original energy itself admits an exact representation by a CSSV-NN.

To get an even better distinction of the expressivity of the three neural network models in the compressible case, one final example is conducted in the following section. 

\begin{table}[h!]
\caption{Achieved mean squared errors of the CSSV-NN, reduced CSSV-NN and Ball NN for the training data sets of the compressible Mielke as well as the additive Mielke-type model.}\label{tab:mses_appendix}
\centering
\begin{tabular}{|l|c|c|c|c|}
\hline
              & CSSV-NN & Reduced CSSV-NN & Ball NN \\ \hline
Mielke & 1.67e-04 & 2.80e-03 & 1.39e-02 \\ \hline
Additive Mielke-type & 3.40e-10 & 1.23e-03 & 1.09e-02 \\ \hline
\end{tabular}%
\end{table}

\begin{figure}[htbp!]
\centering
\caption*{\textbf{Mielke}}
\includegraphics[width=\textwidth]{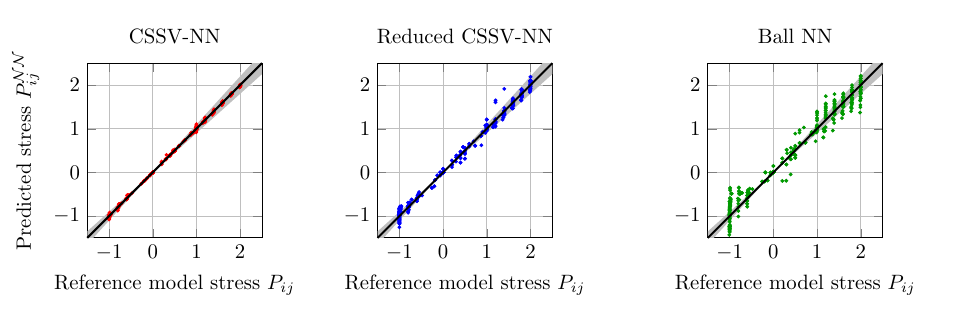}
\caption{Approximation of the Mielke model \eqref{eq:mielke_approx2}: Predicted stress of the CSSV-NN, the reduced CSSV-NN and the Ball NN over the reference model stress. The gray area indicates 10\% error bounds.}
\label{fig_mielke_prediction}
\end{figure}

\begin{figure}[htbp!]
\centering
\caption*{\textbf{Mielke}}
\includegraphics[width=\textwidth]{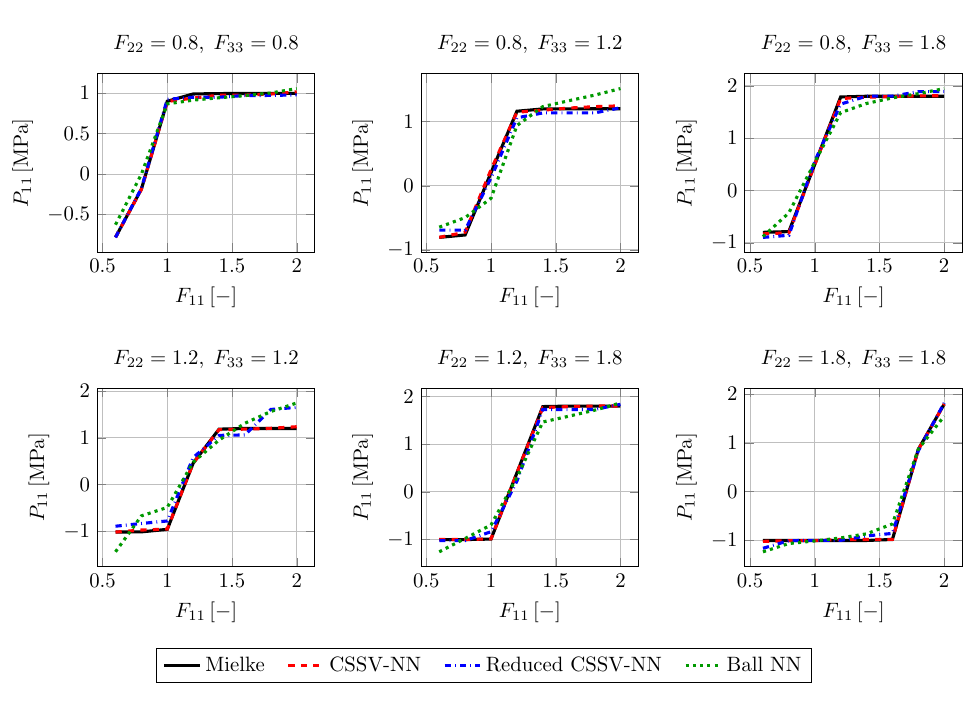}
\caption{Approximation of the Mielke model \eqref{eq:mielke_approx2}: Predicted stress $P_{11}$ of the CSSV-NN, the reduced CSSV-NN and the Ball NN over $F_{11}$ for various, but fixed combinations of $F_{22}$ and $F_{33}$. Plotted are only the training data points.}
\label{fig_mielke_stress11}
\end{figure}

\begin{figure}[htbp!]
\centering
\caption*{\textbf{Mielke}}
\includegraphics[width=\textwidth]{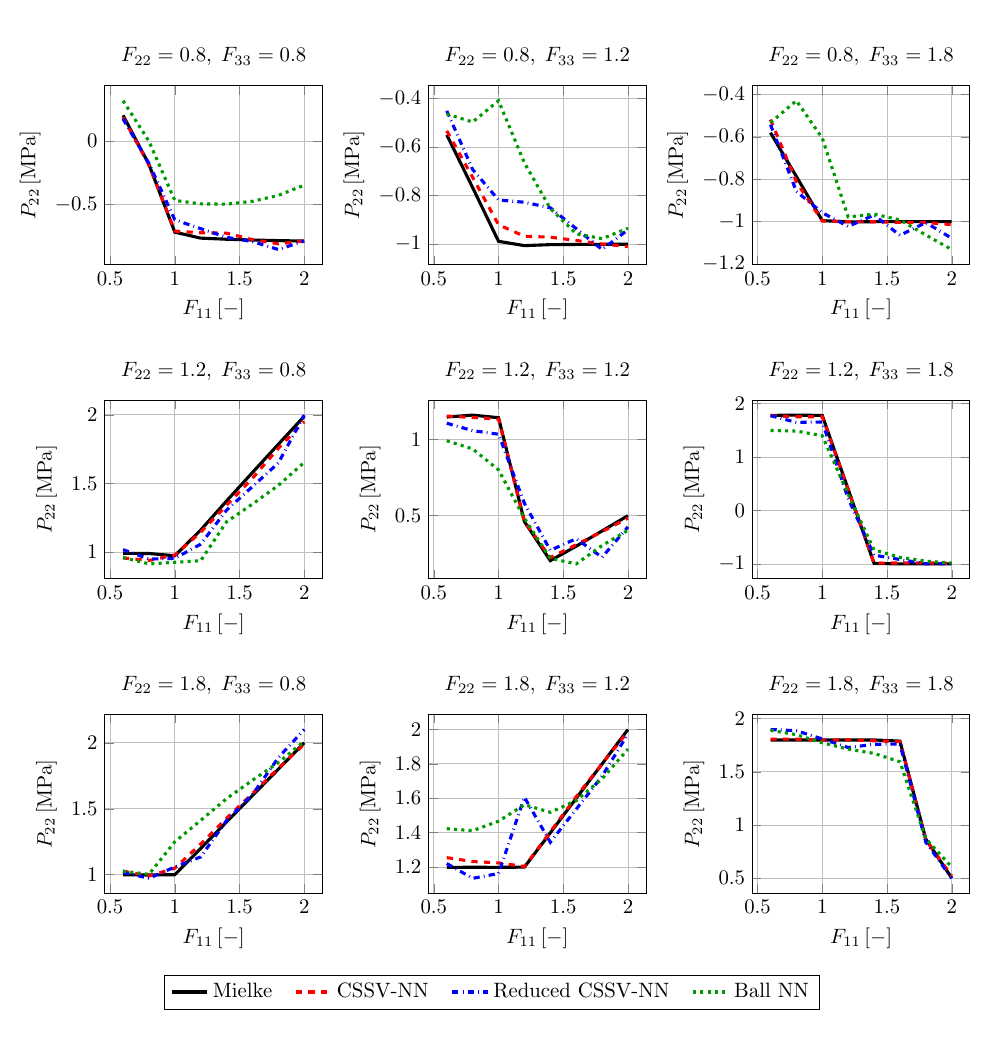}
\caption{Approximation of the Mielke model \eqref{eq:mielke_approx2}: Predicted stress $P_{22}$ of the CSSV-NN, the reduced CSSV-NN and the Ball NN over $F_{11}$ for various, but fixed combinations of $F_{22}$ and $F_{33}$. Plotted are only the training data points.}
\label{fig_mielke_stress22}
\end{figure}

\subsubsection{Compressible additive Mielke-type energy} \label{chp_mielke_sum}
A very similar energy to the one of Mielke is obtained, if the maximum is replaced by the sum of all arguments:
\begin{align}
\Psi_{\textnormal{MT}} = \vert \nu_1 - \nu_2 \nu_3 \vert + \vert \nu_2 - \nu_1 \nu_3 \vert + \vert \nu_3 - \nu_1 \nu_2 \vert. \label{eq:mielke_sum}
\end{align}
This again yields a polyconvex yet non-differentiable energy. A smooth, polyconvex approximation is given by
\begin{align}
\Psi^{\textnormal{approx}}_{\textnormal{MT}} = \dfrac{1}{a} \left( \textnormal{log} \left( \textnormal{cosh} \left( a \, (\nu_1 - \nu_2 \nu_3) \right) \right) + \textnormal{log} \left( \textnormal{cosh} \left( a \, (\nu_2 - \nu_1 \nu_3) \right) \right) + \textnormal{log} \left( \textnormal{cosh} \left( a \, (\nu_3 - \nu_1 \nu_2) \right) \right) \right), \label{eq:mielke_sum_approx} 
\end{align}
wherein $a$ once again controls the smoothness and was set to 10. The dataset was generated in the same way as for the previous example by sampling diagonal deformation gradients of the form
\begin{equation}
\bsym F= \textbf{diag}(\lambda_1,\lambda_2,\lambda_3),
\end{equation}
where each principal stretch $\lambda_i$ was selected from the uniformly discretized range $[0.6:0.2:2.0]$. To avoid redundant permutations, only unique stretch combinations satisfying $\lambda_1 \geq \lambda_2 \geq \lambda_3$ were considered, resulting in 120 distinct deformation states. The training procedure was the same as described in the previous section for the Mielke energy. The UInvar NN was again not investigated here.
The CSSV-NN alone is able to approximate the energy perfectly with a mean squared error of 3.40e-10 (Tab.~\ref{tab:mses_appendix}), cf. Fig.~\ref{fig_mielke_type}. The reduced CSSV-NN only achieved a mean squared error of 1.31e-03 and even exceeds the 10\% error bound for some data points (gray area in Fig.~\ref{fig_mielke_type}). The Ball NN shows large errors for many data points and only achieved a mean squared error of 1.09e-02. The stress-stretch curves (Figs.~\ref{fig_mielke_type_stress11} \& \ref{fig_mielke_type_stress22}) visualize the perfect approximation of the CSSV-NN. Despite the much higher mean squared error compared to the CSSV-NN, the reduced CSSV-NN also captures the material behavior really well with only slight deviations for $P_{11}$ (Fig.~\ref{fig_mielke_type_stress11}) and some smooth approximation of the rapid stress change in 22-direction (Fig.~\ref{fig_mielke_type_stress22} top left and top middle). While the Ball NN captures the stress responses at least qualitatively, it shows some large errors for the $P_{22}$ stresses (Fig.~\ref{fig_mielke_type_stress22}).

These observations once again underline the limited expressiveness of the Ball NN and that the improved version of Ball's polyconvexity criterion is not necessary in the compressible case. Furthermore, at least the results for the Mielke energy indicate that the reduced criterion based on the elementary symmetric polynomials of the signed singular values may also be not necessary and thus the ?-marked relation in Fig.~\ref{fig_crit_overview} seems not to be valid in the compressible case. However, a rigorous mathematical proof is still missing.
\begin{figure}[htbp!]
\centering
\caption*{\textbf{Additive Mielke-type energy}}
\includegraphics[width=\textwidth]{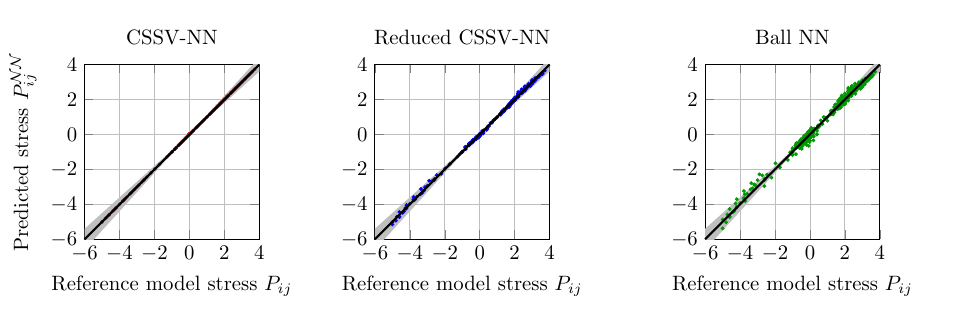}
\caption{Approximation of the additive Mielke-type model \eqref{eq:mielke_sum_approx}: Predicted stress of the CSSV-NN, the reduced CSSV-NN and the Ball NN over the reference model stress. The gray area indicates 10\% error bounds.}
\label{fig_mielke_type}
\end{figure}

\begin{figure}[htbp!]
\centering
\caption*{\textbf{Additive Mielke-type energy}}
\includegraphics[width=\textwidth]{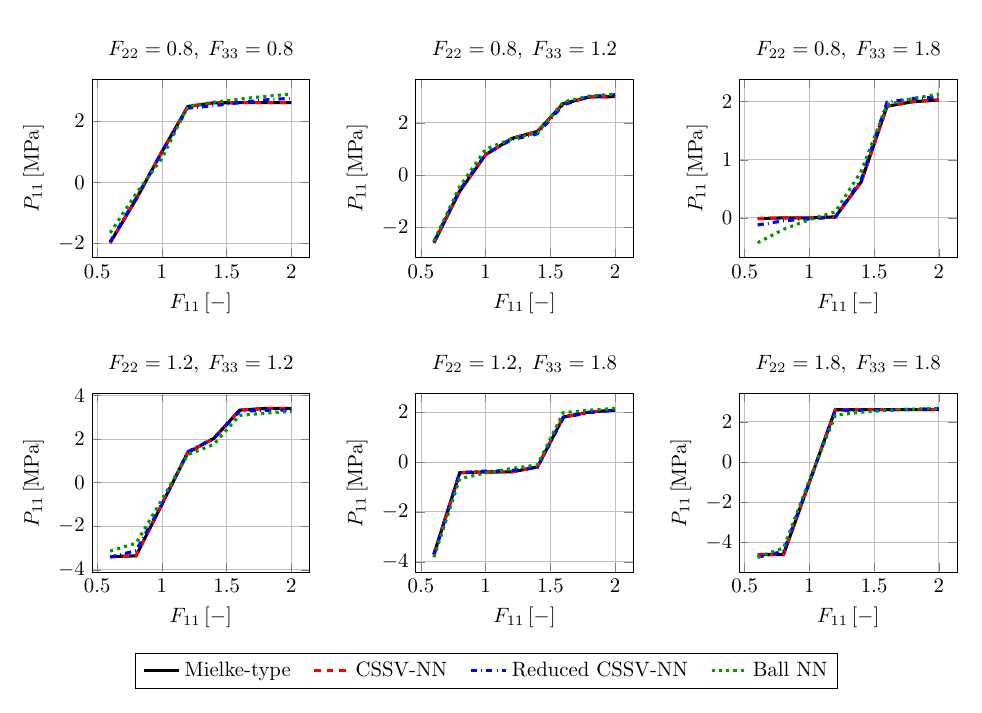}
\caption{Approximation of the additive Mielke-type model \eqref{eq:mielke_sum_approx}: Predicted stress $P_{11}$ of the CSSV-NN, the reduced CSSV-NN and the Ball NN over $F_{11}$ for various, but fixed combinations of $F_{22}$ and $F_{33}$. Plotted are only the training data points.}
\label{fig_mielke_type_stress11}
\end{figure}

\begin{figure}[htbp!]
\centering
\caption*{\textbf{Additive Mielke-type energy}}
\includegraphics[width=\textwidth]{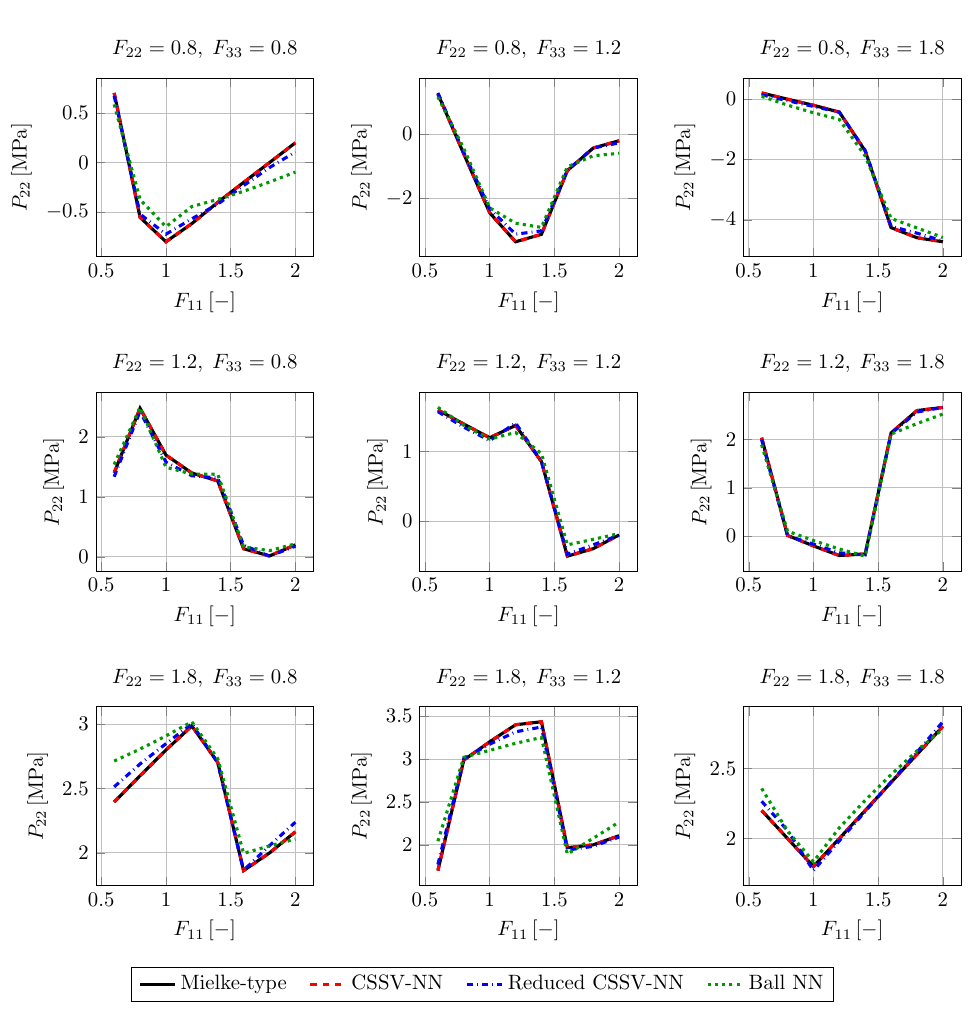}
\caption{Approximation of the additive Mielke-type model \eqref{eq:mielke_sum_approx}: Predicted stress $P_{11}$ of the CSSV-NN, the reduced CSSV-NN and the Ball NN over $F_{11}$ for various, but fixed combinations of $F_{22}$ and $F_{33}$. Plotted are only the training data points.}
\label{fig_mielke_type_stress22}
\end{figure}

\section{Conclusion}

This work provided a comprehensive analysis of sufficient and necessary criteria for frame-indifferent, isotropic polyconvex energies and demonstrated how these criteria directly affect the design and performance of neural network models for hyperelastic materials. By considering both compressible and incompressible formulations, the study also assessed the practical applicability of the different approaches and in the end enables a clear recommendation regarding which models should be used.

The numerical examples show that the inc-CSSV-NN and, in particular, the reduced inc-CSSV-NN reproduce all incompressible material models, the experimental Treloar data and the challenging, incompressible Mielke energy with excellent accuracy. Their ideal performance for every numerical example highlights their high expressiveness and robustness. On the other hand, the reduced CSSV-NN is not able to approximate the compressible Mielke and additive Mielke-type energy perfectly, which indicates that the underlying polyconvexity criterion might not be necessary. However, a formal proof for that remains an open challenge in the compressible and incompressible case, while a proof of the universal approximation theorem for inc-CSSV-NNs was presented. The extension of the CSSV-NN framework to incompressible materials turned out to be especially relevant, as incompressibility plays a central role for many polymers as well as biological materials and can be enforced naturally within this formulation.

In contrast, neural networks based on more restrictive sufficient criteria show clear limitations. The inc-/Ball NN fits classical and experimental data well but fails to approximate the Mielke energy in the compressible as well as incompressible case, substantiating that both -- Ball’s original criterion and the improved version -- are sufficient but not necessary and can significantly restrict expressiveness. In addition, the inc-/Ball NN is computationally more expensive compared to the reduced inc-/CSSV-NN. The inc-UInvar NN performs considerably worse in all settings and does not capture essential characteristics of the underlying energies. Therefore, neither the inc-/Ball NN nor the inc-/UInvar NN should be used when accuracy and flexibility are required.

Altogether, the results demonstrate that neural network models relying on necessary and sufficient criteria for polyconvexity offer clear advantages over classical invariant-based or overly restrictive approaches. Tab.~\ref{tab:overview_models} summarizes the performance of all neural network variants for all investigated material models and Treloar's experimental data. The inc-/CSSV-NN provides the most accurate, expressive and practically applicable framework and is therefore recommended for modeling frame-indifferent, isotropic polyconvex hyperelasticity of compressible and incompressible materials. If high efficiency is required, the reduced inc-CSSV-NN may also be suitable for incompressible materials.

\begin{table}[h]
\caption{Performance of all neural network variants for all investigated incompressible material models, Treloar's experimental data and the compressible Mielke and additive Mielke-type model.}\label{tab:overview_models}
\centering
\renewcommand{\arraystretch}{1.5}

\begin{tabular}{|c|c|c|c|c|}
\hline
  &
 inc-/CSSV-NN &
 Reduced inc-/CSSV-NN &
 inc-/Ball NN &
 inc-/UInvar NN \\ \hline

 Neo--Hooke &
\cellcolor[rgb]{0.20, 0.70, 0.30} $++$ &
\cellcolor[rgb]{0.20, 0.70, 0.30} $++$ &
\cellcolor[rgb]{0.20, 0.70, 0.30} $++$ &
\cellcolor[rgb]{1.00, 0.65, 0.0} $-$ \\ \hline

 Mooney--Rivlin &
\cellcolor[rgb]{0.20, 0.70, 0.30} $++$ &
\cellcolor[rgb]{0.20, 0.70, 0.30} $++$ &
\cellcolor[rgb]{0.20, 0.70, 0.30} $++$ &
\cellcolor[rgb]{0.40, 0.9, 0.40} $+$ \\ \hline

 Gent &
\cellcolor[rgb]{0.20, 0.70, 0.30} $++$ &
\cellcolor[rgb]{0.20, 0.70, 0.30} $++$ &
\cellcolor[rgb]{0.20, 0.70, 0.30} $++$ &
\cellcolor[rgb]{1.00, 0.65, 0.0} $-$ \\ \hline

 Arruda--Boyce &
\cellcolor[rgb]{0.20, 0.70, 0.30} $++$ &
\cellcolor[rgb]{0.20, 0.70, 0.30} $++$ &
\cellcolor[rgb]{0.20, 0.70, 0.30} $++$ &
\cellcolor[rgb]{1.00, 0.65, 0.0} $-$ \\ \hline

 Treloar &
\cellcolor[rgb]{0.20, 0.70, 0.30} $++$ &
\cellcolor[rgb]{0.20, 0.70, 0.30} $++$ &
\cellcolor[rgb]{0.20, 0.70, 0.30} $++$ &
\cellcolor[rgb]{1.00, 0.1, 0.0} $--$ \\ \hline

 inc-Mielke &
\cellcolor[rgb]{0.20, 0.70, 0.30} $++$ &
\cellcolor[rgb]{0.20, 0.70, 0.30} $++$ &
\cellcolor[rgb]{1.00, 0.65, 0.0} $-$ &
\cellcolor[rgb]{1.00, 0.1, 0.0} $--$ \\ \hline

 Mielke &
\cellcolor[rgb]{0.40, 0.9, 0.40} $+$ &
\cellcolor[rgb]{1.00, 0.65, 0.0} $-$ &
\cellcolor[rgb]{1.00, 0.1, 0.0} $--$ &
 \\ \hline

Additive Mielke-type &
\cellcolor[rgb]{0.20, 0.70, 0.30} $++$ &
\cellcolor[rgb]{0.40, 0.9, 0.40} $+$ &
\cellcolor[rgb]{1.00, 0.95, 0.40} o &
 \\ \hline

\end{tabular}\vspace{0.5em}
\caption*{$++:$ Perfect fit, $+:$ Good fit, o$:$ Moderate fit, $-:$ Poor fit, $--:$ Very poor fit}
\end{table}

\section*{Acknowledgements}
Financial support from the German Research Foundation (DFG) via SFB/TRR 188 (278868966), project C01, is gratefully acknowledged. Furthermore, the authors gratefully acknowledge the computing time
provided on the Linux HPC cluster at Technical University Dortmund (LiDO3), partially funded in the course of the Large-Scale Equipment Initiative by the German Research Foundation (DFG) as project 271512359. This work was also funded by the European Union - NextGenerationEU - FINM-MEH - uniri-mzi-25-16. The views and opinions expressed are solely those of the authors and do not necessarily reflect the official stance of the European Union or the European Commission. Neither the European Union nor the European Commission can be held accountable for them.
\begin{figure}[htbp!]
	\centering
	\begin{minipage}{0.35\textwidth}
		\centering
		\includegraphics[width=\linewidth]{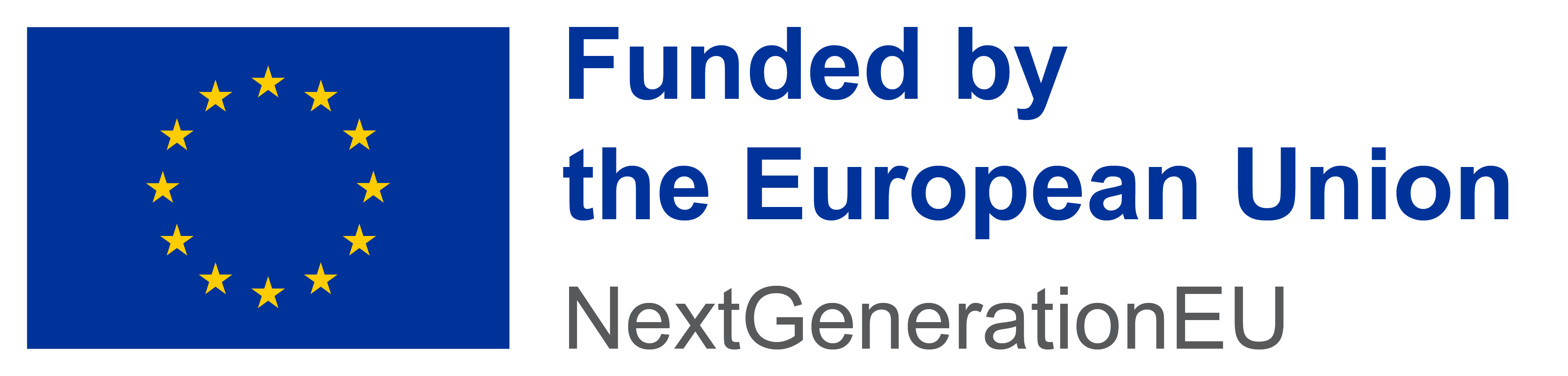}
		\label{fig:subA}
	\end{minipage}
\end{figure}

\bibliographystyle{apalike-ejor}%
\bibliography{Literatur}%

\begin{appendix} 
\section{Proof of Theorem \ref{theo_princstretchpolyconv}: Sufficient polyconvexity criterion based on principal stretches} \label{proof_ball}
\begin{proof}
Let $\Psi$, $\breve\Psi$ and $\Psi_{\textnormal{Ball}}$ be given as in Theorem \ref{theo_princstretchpolyconv}. In order to show the polyconvexity of $\Psi$ using Theorem \ref{theo_singvalpolyconv}, define $\Psi_{\textnormal{ssv}} \colon \R^7 \to \R_\infty$ by
	\begin{equation}
\Psi_{\textnormal{ssv}} (x_1,x_2,x_3,x_4,x_5,x_6,x_7) \coloneqq \begin{cases}
\Psi_{\textnormal{Ball}}(|x_1|,|x_2|,|x_3|,|x_4|,|x_5|,|x_6|,x_7) & \text{if } x_7 \geq 0
\\
\infty& \text{if } x_7 < 0\,.
\end{cases}
\end{equation}
It has to be shown that $\Psi_{\textnormal{ssv}}$ is convex and lower semicontinuous and $\widetilde{\Psi}$ is $\Pi_3$-invariant.
Let $\bsym x,\bsym y \in \R^6 \times[0, \infty)$ and $t \in [0,1]$ be arbitrary, then
\begin{align}
\begin{aligned}\label{eq:134552}
\Psi_{\textnormal{ssv}}( t \, \bsym x +(1-t) \, \bsym y) &= \Psi_{\textnormal{Ball}}( |t \, x_1 +(1-t)\, y_1|, \dots, |t \, x_6 +(1-t)\, y_6|, t\, x_7 +(1-t)\, y_7 ) 
\\
&\leq
\Psi_{\textnormal{Ball}}( t\, |x_1| +(1-t)\, |y_1|, \dots, t\, |x_6| +(1-t)\, |y_6|, t\, x_7 +(1-t)\, y_7 ) 
\\
&\leq t\,\Psi_{\textnormal{Ball}}(|x_1|, \dots,|x_6|, x_7 ) + (1-t)\,\Psi_{\textnormal{Ball}}(|y_1|, \dots,|y_6|, y_7 ) 
\\
&= t\, \Psi_{\textnormal{ssv}}(\bsym x) + (1-t)\,\Psi_{\textnormal{ssv}}(\bsym y) \,,
\end{aligned}
\end{align}
where in the first inequality it is used that the absolute value is convex and $\Psi_{\textnormal{Ball}}$ is non-decreasing in its first six arguments. If $x_7$ or $y_7$ is negative the right-hand side of Eq.~\eqref{eq:134552} 
becomes $\infty$ and the convexity is trivial. Thus, $\Psi_{\textnormal{ssv}}$ is convex.

The function $\Psi_{\textnormal{ssv}}$ is lower semicontinuous on $\R^6 \times[0, \infty)$ because it is the composition of a lower-semicontinuous function $\Psi_{\textnormal{Ball}}$ with the continuous function $|\cdot|$. Since $\R^6 \times[0, \infty)$ is closed, the extension by $\infty$ outside this set preserves the lower semicontinuity.

It remains to show that $\widetilde{\Psi}$ is $\Pi_3$-invariant, i.e.~
\begin{equation}\label{eq:23456}
\widetilde{\Psi}(\bsym \nu) = \widetilde{\Psi}(\bsym{\mathcal{P}} \cdot \bsym \nu) \quad \forall \, \bsym \nu \in \R^3\text{ and } \bsym{\mathcal{P}} \in \Pi_3\,.
\end{equation}
For $\nu_1\nu_2\nu_3 \geq 0$, one can identify $\bsym \nu = \bsym B \cdot \textbf{diag}(\bsym \epsilon) \cdot \bsym \lambda$ and $\bsym{\mathcal{P}} \cdot \bsym \nu = \bsym B' \cdot \textbf{diag}(\bsym \epsilon') \cdot \bsym \lambda$ for some $\bsym B, \bsym B'\in \operatorname{Perm}(3)$ and $\bsym \epsilon, \bsym  \epsilon'\in \lbrace -1,1 \rbrace^3$ with $\epsilon_1 \, \epsilon_2 \, \epsilon_3 = 1$ and
$\lambda_i = |\nu_i|$. Compute 
\begin{align}\label{eq:2456}
	\begin{aligned}
\widetilde{\Psi}(\bsym  \nu) &= \Psi_{\textnormal{ssv}}(m(\bsym \nu)) = \Psi_{\textnormal{Ball}}( |m(\bsym \nu)_1|, \dots,|m(\bsym \nu)_6|, m(\bsym \nu)_7) =
\Psi_{\textnormal{Ball}}( m(\bsym \lambda)_1, \dots,m(\bsym \lambda)_6, m(\bsym \lambda)_7) 
\\
&=
\breve{\Psi}(\bsym B\cdot \bsym \lambda) = \breve{\Psi}(\bsym \lambda) \,,
	\end{aligned}
\end{align}
where in the last equality the $\Perm(3)$-invariance of $\breve{\Psi}$ is used.
This implies Eq.~\eqref{eq:23456} in the case that $\nu_1\nu_2\nu_3 \geq 0$. If $\nu_1\nu_2\nu_3 < 0$ both sides of Eq.~\eqref{eq:23456} are equal to $\infty$. Thus, $\widetilde{\Psi}$ is $\Pi_3$-invariant.

Since $\widetilde{\Psi}$ is $\Pi_3$-invariant, it can be identified with an isotropic and frame-indifferent function. By Eq.~\eqref{eq:23456} and \eqref{eq:6785}, this function can be identified with $\Psi$.
Thus, Eq.~\eqref{eq_psiisoandconv} from Theorem \ref{theo_singvalpolyconv} holds and one obtains the polyconvexity of $\Psi$.
\end{proof}
\end{appendix}
\end{document}